%% file: fullTSP_revised.tex
\documentclass[conference]{IEEEtran}
\IEEEoverridecommandlockouts
\usepackage{cite}
\usepackage{amsmath, amssymb, amsfonts, amsthm, mathtools}
\usepackage{algorithmic}
\usepackage{graphicx}
\usepackage{textcomp}
\usepackage{xcolor}
\usepackage{subcaption, soul}
\def\BibTeX{{\rm B\kern-.05em{\sc i\kern-.025em b}\kern-.08em
    T\kern-.1667em\lower.7ex\hbox{E}\kern-.125emX}}
\usepackage[colorlinks=true, allcolors=blue]{hyperref}
\newtheorem{proposition}{Proposition}[section]

\input macros1

\input exmacros

\renewcommand{\E}{\mathbb{E}}

\begin{document}

\title{\large 
Statistical Analysis of the Extensive Cancellation Algorithm for Passive Radar Using an Imperfect Reference Signal
\thanks{The research was funded by the strategic innovation programme “Smartare Elektroniksystem”, a joint research project financed by VINNOVA, Formas and the Swedish Energy Agency, and by SAAB.}
}

\author{\IEEEauthorblockN{Mats Viberg\thanks{Mats Viberg is also with Chalmers University of Technology.}}
\IEEEauthorblockA{
\textit{Blekinge Institute of Technology}\\
Karlskrona, Sweden 
}
\and
\IEEEauthorblockN{Daniele Gerosa}
\IEEEauthorblockA{
\textit{Chalmers University of Technology}\\
Göteborg, Sweden 
}
\and
\IEEEauthorblockN{Tomas McKelvey}
\IEEEauthorblockA{
\textit{Chalmers University of Technology}\\
Göteborg, Sweden 
}
\and
\IEEEauthorblockN{Thomas Eriksson}
\IEEEauthorblockA{
\textit{Chalmers University of Technology}\\
Göteborg, Sweden 
}
}

\maketitle

\begin{abstract}

Passive radar systems have received tremendous attention over the past few decades, due to their low cost and ability to remain covert during operation. Such systems 
rely on a so-called Illuminator-of-Opportunity (IO), for example, a commercial TV station. We consider a network of Receiving Nodes (RN) without spatial resolution capability, which receives the direct signal and reflections from both stationary objects (clutter) and possible targets. After suitable preprocessing, the RNs transmit information to a Fusion Center (FC) that performs the final target detection, localization and tracking.

Several methods for target localization have been proposed in the literature, and our
focus is on the seminal Extensive Cancellation Algorithm (ECA)
\cite{colone_multistage_2009}. In this approach,
each RN collects information about target parameters, while
canceling interference using a projection.
This is done by exploiting a separate Reference Channel (RC), which captures the IO signal without interference apart from receiver noise.

We derive the statistical properties of the ECA parameter estimates under the assumption of a high Signal-to-Noise Ratio (SNR), and we give a sufficient condition for the SNR in the RC to enable statistically efficient estimates. The theoretical results are corroborated through computer simulations, which indicate that the theory agrees well with empirical results
under practical operating conditions.
The contributions of this paper can be used, for example, to design experimental setups for feasibility studies and to inform system design for achieving a desired localization accuracy.

\end{abstract}



\section{Introduction}
\label{sec:intr}

The idea of employing illuminators of opportunity (IOs) like frequency modulation (FM) radio or television stations to perform radar tasks is not new, and it goes back at the very least to the successful experiment carried out by sir Robert Watson-Watt in \(1935\), who used the shortwave BBC Empire transmitter at Daventry (England) to detect a slow bomber aircraft at short distance \cite{Kuschel_NATO_2017}. 
Today, passive radar is a fairly well established technology \cite{kuschel_tutorial_2019,Colone_etal_tutorial_MAES2023}, and it 
has found plentiful use in both military and civilian applications, see e.g.
\cite{Kuschel_NATO_2017,Micro_Ukr_2019,OlsenKuschel:MAES2017}.
There are two major differences between passive and active radars. Firstly, and most importantly, in passive radar systems the transmitted waveform is not known at the Receiving Nodes (RNs). Secondly, the waveform is not designed for radar applications, but more commonly for broadcasting information. It can therefore be viewed as a random signal rather than a periodically repeated waveform, which is typical of active radars. Since the waveform is unknown, each RN needs to receive the transmitted waveform through a line-of-sight channel termed the \emph{direct path}. If only a single receiving channel is available, this can be used to calculate the Time-Difference-of-Arrival (TDOA) between the direct path and potential targets.
It is well-known that noise in the reference signal (the direct path) can severely deteriorate the detection capability and interference rejection, see e.g. \cite{liu_hongbing_himed_2014,Mahfoudia_etal_RefSigReconstruct:RadCon2017,Zhou_etal:DirectLocalizationReferenceChannelFreePassiveRadar_TSP2025}. 
Therefore, most practical systems are equipped with a separate Reference Channel (RC) with a high-gain antenna directed towards the IO. 
The RC is usually assumed to have negligible contributions from clutter and target reflections, and this is also the case considered herein. In effect, the RC serves as a reference for extracting information from the Surveillance Channel (SC), which receives the target signal corrupted by direct path and clutter interference. 

In the present work, we are mainly concerned with the problem of estimating the target parameters in the presence of interference and using an imperfect reference channel.
Several algorithms have been proposed in the literature for solving this problem, including
\cite{colone_multistage_2009,PalmerSearle_Adaptive_Filter_IEEE_Radar_Conf_2012,Ma_etal_SparseAdaptFiltNLMS_GRSLett2016,bai_etal_2021,LyuDing:_IET2022,Nouar_etal_Piers:2023,zhou_direct_2024,Guo_etal_ClutterSuppress_DSP:2025}.
A recent contribution, \cite{Zhou_etal_HybridActPassRadar_TAES2025} considers the interesting possibility to combine the active and passive modes.
Target parameter estimates
may later be leveraged in a Generalized Likelihood Ratio Test (GLRT) for solving the detection problem, see e.g. \cite{thomopoulos_optimal_1989,tajer_optimal_2010,yi_suboptimal_2020}.

Despite the rich efforts regarding estimation methods and their efficient implementation (e.g. \cite{Colone_etal_SlidingECA_TAES2016,Tang_etal_FastECA_TAES2025}), the quality of the resulting parameter estimates and its dependence on various system parameters is not well understood. In particular, the effect of noise in the reference signal needs to be investigated in more detail, since it leads to both imperfect interference cancellation and mismatched filtering.
Further, the question of identifiability of the target parameters at each RN has also not been addressed for a general waveform.
The purpose of this paper is to fill these gaps by 1) deriving an explicit formula for the
Mean Squared Error (MSE) of the target parameter estimates under the above mentioned scenario, and 2) proving that the target parameters can be uniquely recovered  
for IO signals belonging to a certain class of stochastic processes.

For a perfectly known reference signal, the target parameter estimates are expected to attain the Cramér-Rao Lower Bound (CRLB) on the MSE, derived for such a favorable case
\cite{he_noncoherent_2010,greco_cramer-rao_2011,Liu_etal_Survey_Limits_ISAC:2022}. 
In \cite{tong_cramerrao_2019}, the CRLB is derived for the case where an RC is unavailable and the reference signal is estimated jointly with the target parameters, while \cite{Zhou_etal_HybridActPassRadar_TAES2025} considers hybrid active/passive radar systems. Neither of these include the effect of clutter interference, though.
Although the referenced works can shed light on how localization accuracy depends on the scenario parameters, they do not consider the (theoretical) performance of an actual estimator. Our statistical analysis extends the corresponding results for ``classical" radar and array signal processing, such as \cite{StoicaN:89},
by considering widely separated antennas and an uncertain reference signal. A somewhat related performance analysis is found in \cite{TirerWeiss_TSP2017}, where Direct Position Determination (DPD, see below) based on time-delay data only is studied, assuming antenna arrays at each RN. The study perhaps closest to the one presented here is
\cite{wang_statistical_2017}, which analyzes DPD based on single sensor Doppler-shift measurements without interference. The referenced work includes the effect of modeling errors, which makes the analysis somewhat related to ours. However, the estimator in \cite{wang_statistical_2017} has a different structure, and the question of uniqueness is not posed.

In the present paper, we focus on the Extensive Cancellation Algorithm (ECA) \cite{colone_multistage_2009}, as a state-of-the-art algorithm for the studied scenario in terms of statistical accuracy. In this method, the interference from the direct path and clutter is modeled by a linear Finite Impulse Response (FIR) filter, different at each RN. Thus, the interference resides in a subspace that corresponds to stationary reflections from a set of ``contaminated" range gates. The ECA algorithm first applies a projection operator 
to the SC data to remove the interference.
Subsequently, a matched-filter is applied to identify the target range (time-delay) and radial speed (Doppler) parameters. 
We present the ECA algorithm in the framework of Maximum Likelihood (ML) estimation, and show that it can be interpreted as an approximate ML estimator.
The approximation consists of using the data from the RC channel as if it were the true transmitted waveform.
This is also in line with \cite{zhou_direct_2024}, where, however, clutter in the SC is not considered.
Since the data at each RN is independent, the global likelihood function is calculated at the FC, simply by adding the contributions (log-likelihood) from each node. This is now viewed as a function of the target position and velocity, and its maximum with respect to these parameters yields the final estimate. This procedure is often referred to as Direct Position Determination (DPD) \cite{WeissDPD_Book2009,BarShalomWeiss_SP2011}. While it requires significant communication overhead and solving a complicated optimization problem, DPD may yield superior performance as compared to fusing only local estimates from the RNs, as in \cite{MalanowskiMK:2012,Zhang_etal_Localization_TSP2025}. 
Yet, as pointed out in \cite{wang_statistical_2017},
the latter two-step approach can achieve the same asymptotic performance as DPD, by exploiting the Extended Invariance Principle (EXIP) of \cite{stoica_reparametrization_1989}. This requires, though, that the SNR is high enough at all RNs, while DPD only needs the "total SNR" to be high. See also \cite{Zhou_etal_HybridActPassRadar_TAES2025}, which includes a comparison of the two approaches.

The purpose of the present contribution is to study the estimation accuracy rather than to propose a new algorithm, and we focus on the ECA-based DPD approach as described above.
The novel contributions are summarized as follows:
\begin{enumerate}
    \item An interpretation of the Extensive Cancellation Algorithm (ECA) of \cite{colone_multistage_2009} as an approximate Maximum Likelihood technique.
    \item A proof of parameter identifiability and consistency of the parameter estimates for a class of stochastic processes. 
    \item A high-SNR analysis of the ECA estimates of the delay/Doppler parameters and/or the target location and velocity. An explicit expression of the covariance matrix of the estimated target parameters is given, and it is interpreted as the sum of the CRLB 
    assuming a perfect knowledge of the transmitted waveform and the excess contribution due to noise in the reference channels.
    \item A compact expression for the CRLB that involves only the parameters of interest (delay and Doppler), and not other unknowns related to amplitude or direct path and clutter interference.
    \item Based on the above, explicit guidelines for the SNR in the RC in order for the parameter estimates to achieve the mentioned CRLB. This is useful to set requirements on, for example, the antenna gain and the receiver electronics such that noise in the RC does not significantly deteriorate the estimation of the target parameters.
\end{enumerate}
We remark that the results on parameter identifiability and the compact CRLB formulation is applicable also to the case where the transmitted waveform is perfectly known, referred to as \emph{noise radar} \cite{Savci_etal_NoiseRadarOverview_MAES2020,galati_signal_2022,AnkeletalNoiseRadarIET2023}.

\subsection{Symbols and notation}
We indicate \emph{column} vectors with bold lowercase letters \( \mathbf{a}, \mathbf{b}, \dots \) and matrices with bold uppercase letters \( \mathbf{A}, \mathbf{B}, \dots \). The element-wise Hadamard product is denoted by \( \odot \) and the norm \( \| \cdot \| \) is the standard Euclidean norm unless otherwise specified.

\section{System Model and Problem Formulation}
\label{sec:prob}

We consider a passive radar scenario, where $K$ stationary RNs collect data emanating from $N_I$ IOs. 
See Figure \ref{fig:geometry} for a scenario with one IO and one RC node.
\begin{figure}
\includegraphics[width=0.5\textwidth]{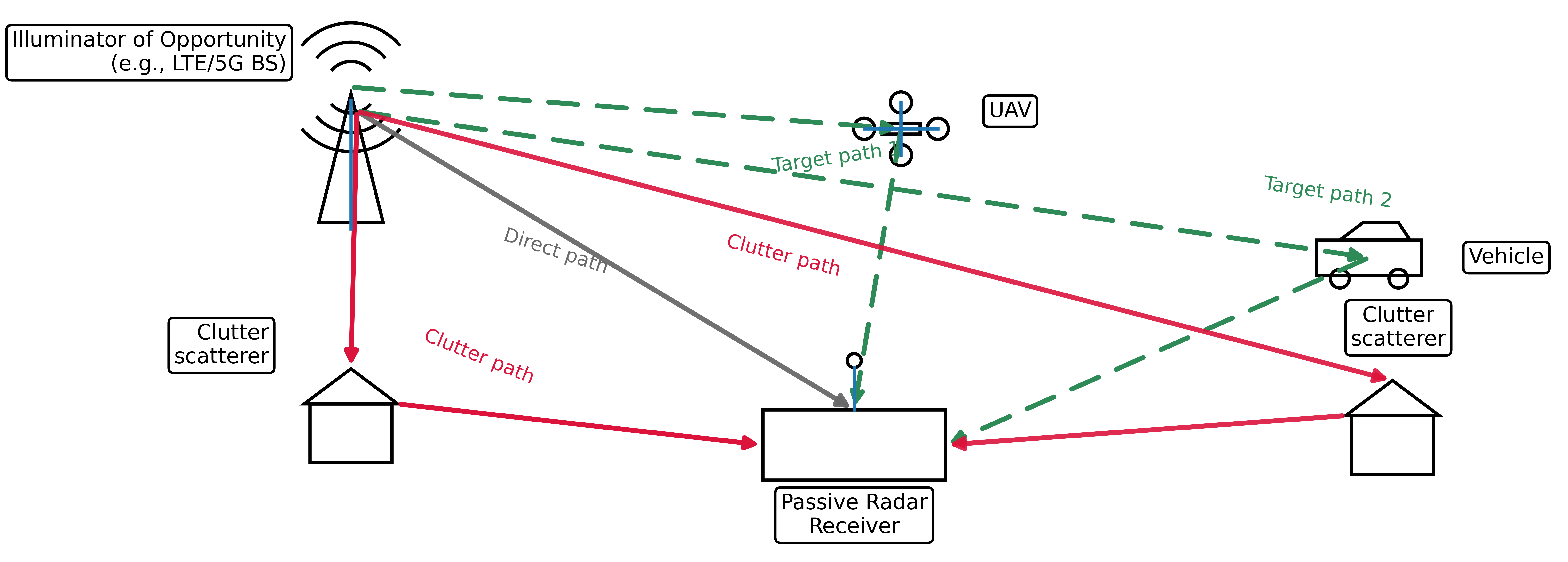}
\caption{General Passive Radar Geometry}
\label{fig:geometry}
\end{figure}
The nodes receive a direct path from the IO, unwanted reflections from stationary objects (clutter) as well as reflections from $N_T$ moving targets.
Each RN is equipped with two antennas; the reference channel (RC) and the Surveillance Channel (SC), respectively. The RC has a directive antenna that captures the direct path at a high SNR,
whereas the SC antenna senses the entire area where targets can appear, thus also receiving both Direct Path Interference (DPI) and Clutter Interference (CI). Secondary reflections are ignored, and the same for clutter and target reflections into the RC.
For simplicity, we assume $N_I=1$, i.e. only one IO is considered. 
Likewise, we consider only the case $N_T=1$, although the studied method can in principle be applied to multiple sufficiently well-separated targets by searching for all local maxima of the global likelihood function to be defined. Indeed, such a procedure is standard in most radar systems to cope with multi-target scenarios, and a separate hypothesis test is then performed at each peak to decide whether a target is present or not. In the present contribution we do not address the detection problem, and assume knowledge of the presence of exactly one target.

In the studied scenario, the following assumptions are made (similar to \cite{colone_multistage_2009}):
\begin{itemize}
    \item The target of interest is moving at a constant non-zero velocity and the data collection time is short enough for the scenario parameters to remain approximately constant.
    \item The interference contributions to the
    SC from clutter and the direct path are stationary, denoted CI and DPI respectively.
    \item The RC uses a directive antenna, so that the clutter and target contributions to the RC can be ignored in relation to the direct-path IO signal.
    \item The RNs can transmit data to a Fusion Center (FC) after appropriate pre-processing. The data from different receiver nodes are synchronized to time-delay (relative bandwidth) but not to phase (relative carrier).
    \item The positions of the IO and the RNs are known at the FC. For simplicity, we consider a 2D geometry with spatial coordinates $x$ and $y$. The extension to 3D comes with a higher computational demand, but is otherwise straightforward.
    \item The RC and SC receivers at each node are considered co-located and they share the same oscillator and thus experience the same phase noise, which is therefore neglected. 
\end{itemize}

The IO transmits a signal $\Re\{s(t)e^{j\omega_{c}t}\}$, where $s(t)$ is the baseband complex envelope and $\omega_c$ the carrier frequency. 
Given the description and assumptions in the previous section, the available data after demodulation at node $k,\ k=1,\dots,K$, is modeled by
\begin{align}
    x_{k}(t) &= a_{k} s_k(t) + n_{k}(t) \label{eq:ref} \\
    y_{k}(t) &= b_{k} s_k(t) +  
    s_k^c(t)
    + d_{k} s_k(t-\tau_{k}) e^{j\omega_k t} + e_{k}(t)\,, \label{eq:surv}
\end{align}
where $x_k(t)$ and $y_k(t)$ represent the data from the reference and the surveillance channel respectively.
For ease of notation, we define the time reference for RN $k$ at the node itself, so that
$s_k(t)$ is the delayed IO waveform at node $k$. The delay is known at the FC, since the position of the IO and the nodes are assumed known.
The RC signal $x_k(t)$ in
(\ref{eq:ref}) consists of a direct path, where $a_{k}$ accounts for propagation attenuation and receiver characteristics, and receiver noise $n_{k}(t)$. In the SC (\ref{eq:surv}), the first term is the DPI, the second represents CI,
the third is the target reflection, and the last term is receiver noise. The complex target amplitude $d_k$ accounts for propagation attenuation, bi-static Radar Cross Section (RCS) as well as antenna and receiver gain. The factor $e^{j\omega_k t}$ represents the Doppler effect due to the 
bi-static range rate, as seen from RN $k$ (see below). 
The constant target velocity is assumed to be
sufficiently low so that the geometry of the bi-static scenario does not change significantly over the data collection interval. From Figure \ref{fig:geometry} (see also, e.g. \cite{DuW:2014}), the time-varying time-delay of the target signal can be expressed as
\begin{equation}
\label{eq:tau}
\tau_k(t) = \frac{\|\ubf - \rbf_k\| + \|\ubf - \rbf\|}{c} - \frac{\|\rbf_k - \rbf\|}{c} \, ,
\end{equation}
where $\ubf=(x,y)^T$ is the (time-varying) target position vector, $\rbf_k$ is the location of the $k$th RN, $\rbf$ is the IO position, and $c$ is the speed of propagation. Notice that the line-of-sight delay has been subtracted in \eqref{eq:tau} to comply with the model in \eqref{eq:surv}. The distances from target to RN $k$ and from target to IO are, respectively, denoted as
$r_k = \|\ubf - \rbf_k\|$ and $r = \|\ubf - \rbf\|$. The time derivative of $\tau_k(t)$ due to a constant velocity $\dot{\ubf}=(v_x,v_y)^T$ is given by
\begin{equation}
\label{eq:taudot}
\dot{\tau}_k = \frac{1}{c} \left(
\frac{\dot{\ubf}^T(\ubf - \rbf_k)}{r_k} +
\frac{\dot{\ubf}^T(\ubf - \rbf)}{r}
\right)\, ,
\end{equation}
and is assumed to be constant.
We denote the delay of the target at time $t=0$ as
$\tau_k(0)=\tau_k$, so that
the time-varying time-delay is given by
$
    \tau_k(t) = \tau_k + \dot{\tau}_k t
$. The Doppler shift comes from the time-delay of the carrier and we have $e^{j\omega_{c}(t-\tau_{k}(t))}=e^{j\omega_{c}t}e^{-j\tau_{k}}e^{-j\omega_{c}\dot{\tau}_k t}$, and thus the target Doppler frequency $\omega_k$ is related to $\dot{\tau}_k$ by 
\begin{equation}
    \omega_k = -\dot{\tau}_k \omega_c\, .
    \label{eq:Doppler}
\end{equation}
It is clear that each RN can only determine the time-delay and Doppler of the target as seen from its own position. Given $(\tau_k,\omega_k)$ pairs from sufficiently many RNs, there is a unique target position and velocity vector that satisfies (\ref{eq:tau}) -- (\ref{eq:Doppler}) for all $k$. The general problem under consideration herein is to estimate $\ubf$ and $\dot{\ubf}$ from measured data at the RC and SC at $K$ nodes. 

Let the data collection time in the SC be $0\leq t < T$, in which $N$ samples are collected at time instances $t_n = n\Delta T$, $\Delta T = T/N$, $n=0,\dots,N-1$. 
The sampling time $\Delta T$ satisfies $\Delta T \leq 1/B$, where $B$ is the bandwidth of $s_k(t)$.
If $\dot{\tau}_k$ is much smaller than $\Delta T/T=1/N$, the time-delay in $s(t-\tau_k(t))\approx s(t-\tau_{k})$ is approximately the same for all samples. If this is not the case, the so-called \textit{range cell migration} may occur, so that the discrete-time delay (corresponding to bistatic target range gates) is not constant during the data collection interval. See, e.g. \cite{barott_single-antenna_2014,Martelli_etal:2018} for details and mitigation for the monostatic case. 
For simplicity, we assume here that \eqref{eq:surv} holds for $0\leq t \leq T$, where $\tau_k$ is constant. It is possible to extend the theoretical results to a general case where $\tau_k$ is time-varying, but to keep the exposition comprehensive we do not pursue this any further.
Due to the time-delay of the clutter and target components, it is assumed that the data collection starts earlier for the RC than for the SC. Let the maximum delay of interest be $M$ samples. 
The available data samples are then
$\{x_{k}(t_n)\}_{n=-M}^{N-1}$ and
$\{y_{k}(t_n)\}_{n=0}^{N-1}$, which
can be put into vector form as
\begin{align}
    \xbf_{k}^M &= a_{k} \sbf_k^{M} + \nbf_{k} \label{eq:vref} \\
    \ybf_{k} &= b_{k} \sbf_k + 
    \sbf_k^c +
    d_{k}\, \sbf_k(\tau_{k}) \odot \vbf(\omega_k) + \ebf_{k}\,, \label{eq:vsurv}
\end{align}
where $\xbf_{k}^M = 
[x_{k}(t_{-M}),\dots,x_{k}(t_{N-1})]^T$ is the observed RC signal vector and 
$\sbf_{k}^{M}$ is the vector of received waveforms at the same sample instants,
whereas $\ybf_{k} = 
[y_{k}(t_{0}),\dots,y_{k}(t_{N-1})]^T$ and 
$\sbf_{k}(\tau_k) = 
[s_{k}(t_0-\tau_k),\dots,s_{k}(t_{N-1}-\tau_k)]^T$ are the corresponding SC signals, with the simplified notation $\sbf_k=\sbf_k(0)$ for the DPI waveform, which has zero delay by definition. 
Further, for the target signal, $\vbf(\omega_k)$ denotes the DFT vector
\begin{equation} \begin{split}
    \vbf(\omega_k) & = [e^{j\omega_k t_0}, \dots,e^{j\omega_k t_{N-1}}]^T \\ & 
    =
    [1,e^{j\omega_k \Delta T},  \dots,e^{j\omega_k (N-1) \Delta T}]^T.
    \end{split}
    \label{eq:vdef}
\end{equation}
The clutter component is assumed to be well-modeled as a linear combination of a finite set of past samples of the IO signal, i.e. a FIR filter:
$$
s_k^c(t_n) = \sum\limits_{l\in \mathcal{L}_c}
c_l s_k(t_n - l\Delta T)
=\sum\limits_{l=1}^L
c_l s_k(t_n - l\Delta T)\,,
$$
where the set $\mathcal{L}_c$, of cardinality $L$, $L\leq M$, contains the range of time delays $l$ for which we expect clutter returns. For the sake of simplicity, we assume in the second equality that $\mathcal{L}_c=\{1,\dots,L\}$.
Hence, we can express the clutter vector in (\ref{eq:vsurv}) as
\begin{equation}
    \sbf_k^c = \Sbf_k \cbf_k\,,
\end{equation}
where $\Sbf_k$ is an $N\times L$ Toeplitz matrix containing samples of $s_k(t_n - l\Delta T)$ for 
$l=1$ and $n=0,\dots,N-1$ as its first column, and for $n=0$ and $l=1,\dots,L$ as its first row; and where $\cbf_k = [c_1 ^{(k)},\dots,c_L ^{(k)}]^T$ is the vector of complex-valued FIR filter coefficients. 
The noise contributions are modeled as Additive White Gaussian Noise (AWGN) 
with variances $\sigma_n^2$ and $\sigma_e^2$ respectively. Thus, $\nbf_{k}$ and $\ebf_{k}$ are zero-mean independent Gaussian random vectors with covariance matrices $\sigma_n^2\Ibf_{M+N}^{}$ and $\sigma_e^2\Ibf_N^{}$ respectively, where the dimension of the respective identity matrix has been stressed for clarity.

\section{The ECA Approach to Target Parameter Estimation and Localization}
\label{sec:method}

In this section, we present the ECA approach of \cite{colone_multistage_2009} within a Maximum Likelihood (ML) framework. We stress that the purpose of this paper is not to present a new algorithm, but to derive the statistical properties of the resulting target parameter estimates.
At each RN, the ECA approach cancels the interference in the SC (DPI and CI) by projecting \eqref{eq:vsurv} onto a subspace that is orthogonal to the interference. Subsequently, the 2D Cross-Ambiguity Function (CAF), which is proportional to the negative log-likelihood, is computed to extract information about the target delay and Doppler parameters. Both the interference cancellation and the matched filtering is performed using a noisy estimate of the IO waveform, obtained from the RC.

Estimation of the target position and velocity is performed at the FC by combining the information received from each node. Following the ML framework, we consider the Direct Position Determination (DPD) approach, where the local log-likelihood function of each node is transmitted to the FC. In a practical implementation, the information needs to be compressed before transmission, and ultimately only the local delay-Doppler estimates are used. At high enough SNR, such a two-step approach may be close to optimal, and our statistical analysis could then provide a useful performance prediction for this approach as well. However, a more precise performance comparison is beyond the scope of this paper.

\subsection{Interference Cancellation and Delay-Doppler Estimation}
\label{sec:IC}

Under assumptions of AWGN and modeling the IO signal vector $\sbf_k$ as deterministic and unknown \cite{liu_glrt_2015,zhang_maximum_2019,zhou_direct_2024}, the negative log-likelihood function (ignoring constants) at RN $k$ is given by
\begin{equation}
    \frac{1}{\sigma_n^2} \left\|
\xbf_{k}^M - a_{k} \sbf_k^{M} \right\|^2 + \frac{1}{\sigma_e^2} \left\|
    \ybf_{k} - 
    \Sbf_{I,k}\fbf_k -
    d_{k}\, \sbf_k(\tau_{k}) \odot \vbf(\omega_k) \right\|^2 ,
    \label{eq:ML1}
\end{equation}
where we have collected the interference contribution to the SC in the matrix $\Sbf_{I,k}=[\sbf_k\ \Sbf_k]$, with amplitude vector $\fbf_k = [b_k^{}\ \cbf_k^T]^T$.

The local ML estimate is now found by minimizing (\ref{eq:ML1}) with respect to the unknown parameters that are specific to node $k$. The delay $\tau_k$ and Doppler $\omega_k$ parameters are functions of the target position and velocity according to (\ref{eq:tau})--(\ref{eq:Doppler}), and are therefore shared among the different nodes.
Both the RC and the SC contain information about the unknown IO signal, through the first and second terms in (\ref{eq:ML1}) respectively. Assuming the SNR of the IO signal to be much stronger in the RC than in the SC, due to line-of-sight and directive antenna, it is a reasonable approximation to determine the IO signal using the RC data only \cite{zhou_direct_2024}. This will also significantly simplify the estimation, since the first term in (\ref{eq:ML1}) can be put to zero, for example, by absorbing the amplitude into the estimate of $\sbf_k^{M}$ and taking $\hat{\sbf}_k^{M} = \xbf_k^M$. Using this approach, the noisy signal from the reference channel is used in lieu of the unknown IO waveform when processing the data from the surveillance channel, in accordance with \cite{colone_multistage_2009}. 

We remark that other approximations for reconstructing the reference signal are possible, such as the eigenvector approach of e.g. \cite{zhou_direct_2024}, or by exploiting knowledge of the signal structure. The analysis to be presented may also be applicable to these cases by replacing the RC noise term with the error of the so reconstructed reference signal.

Motivated by high SNR in the RC, we use
$\hat{\sbf}_k^M=\xbf_k^M$ from the first term of \eqref{eq:ML1} 
and rearrange the samples to form
\begin{equation}
\hat{\Sbf}_{I,k} = [\hat{\sbf}_k\ \hat{\Sbf}_k] = [\xbf_k\ \Xbf_k] = \Xbf_{I,k}\, ,
\label{eq:subest}
\end{equation}
where the sampling instances in $\xbf_k \in \mathbb{C}^N$ and $\Xbf_k \in \mathbb{C}^{N\times L}$ are aligned with those in $\sbf_k$ and $\Sbf_k$, and they are all available in $\xbf_k^M \in \mathbb{C}^{M+N}$.
Further, we take $\hat{\sbf}_k(\tau_k)=\xbf_k(\tau_k)$, which is calculated using interpolation when $\tau_k$ is not on the sampling grid.
Inserted into (\ref{eq:ML1}), this yields
\begin{equation}
 \ell(\fbf_k,d_k,\tau_k,\omega_k) = \left\|
    \ybf_{k} - 
    \Xbf_{I,k}\fbf_k -
    d_{k}\, \xbf_k(\tau_{k}) \odot \vbf(\omega_k) \right\|^2 .
    \label{eq:ML2}
\end{equation}
Note that since $\xbf_k \approx a_k\sbf_k$ contains the RC signal amplitude, the estimated amplitude parameters $b_k$, $\cbf_k$ and $d_k$ will in effect be normalized with respect to $a_k$. This normalization does not affect the estimation of the target position and velocity, but it must be taken into account in the statistical analysis as well as in an eventual target detection decision.

Minimizing (\ref{eq:ML2}) w.r.t. 
$\fbf_k$ and substituting the resulting estimate back into (\ref{eq:ML2}), results in an effective cancellation of the DPI and CI. 
To this end, the orthogonal projection matrices onto the span of $\Sbf_{I,k}$ and $\Xbf_{I,k}$ and their respective orthogonal complements are introduced as
\begin{align}
    {\Pibf}_k^{\perp} &=  \Ibf - {\Pibf}_k
    =  \Ibf -
    \Sbf_{I,k}\left(\Sbf_{I,k}^{H}\Sbf_{I,k}^{}\right)^{-1}\Sbf_{I,k}^H
    \label{eq:Piperp} \\
    \hat{\Pibf}_k^{\perp} &=  \Ibf - \hat{\Pibf}_k
    =  \Ibf -
    \Xbf_{I,k}\left(\Xbf_{I,k}^{H}\Xbf_{I,k}^{}\right)^{-1}\Xbf_{I,k}^H   .
    \label{eq:Piperphat}
\end{align}
Note that the amplitude scaling in $\Xbf_{I,k}$ has no effect on the projection matrix, so if \( \sigma_n ^2 = 0 \) we have $\hat{\Pibf}_k={\Pibf}_k$.

With these definitions, the minimum of (\ref{eq:ML2}) w.r.t. $\fbf_k$ reduces to the interference-cleaned criterion function
\begin{equation}
    \ell(d_k,\tau_k,\omega_k) =
     \left\| \hat{\Pibf}_k^{\perp} \bigl( \ybf_k -
    d_{k}\, \xbf_k(\tau_{k}) \odot \vbf(\omega_k) \bigr) \right\|^2 .
\label{eq:ML3}
\end{equation}
Next, we introduce the 
delay-Doppler ``steering vector" $\abf(\tau_k,\omega_k)$, along with its estimate $\hat{\abf}(\tau_k,\omega_k)$, obtained from the RC data as
\begin{align}
\label{eq:adef}
 \abf(\tau_k,\omega_k) &=
\sbf_k(\tau_{k}) \odot \vbf(\omega_k) \\
 \hat{\abf}(\tau_k,\omega_k) &=
\xbf_k(\tau_{k}) \odot \vbf(\omega_k)\,.
\label{eq:ahatdef}
\end{align}
We remind that the amplitude $a_k$ is present in (\ref{eq:ahatdef}) but not in (\ref{eq:adef}).
Substituting the minimizing $d_k$ from (\ref{eq:ML3}) back into the criterion then results in the final form
\begin{equation}
    \ell(\tau_k,\omega_k) = 
    \left\| \hat{\Pibf}_k^{\perp}\ybf_k  -
\hat{\Pibf}_k^{\perp}\frac{\hat{\abf}(\tau_k,\omega_k) \hat{\abf}^H(\tau_k,\omega_k)}{\hat{\abf}^H(\tau_k,\omega_k)\hat{\Pibf}^{\perp}\hat{\abf}(\tau_k,\omega_k)}\, \hat{\Pibf}_k^{\perp}\ybf_k\right\|\, .
\end{equation}
Clearly, minimizing $\ell(\tau_k,\omega_k)$ is equivalent to maximizing the following interference-canceled and normalized version of the 
cross ambiguity function, or ``spectrum":
\begin{equation}
P_k(\tau_k,\omega_k) =
\frac{ |\hat{\abf}^H(\tau_k,\omega_k) \hat{\Pibf}_k^{\perp} \ybf_k|^2}{\hat{\abf}^H(\tau_k,\omega_k)\hat{\Pibf}_k^{\perp} \hat{\abf}(\tau_k,\omega_k)} = 
\left\| \hat{\Pbf}_k
\, \ybf_k \right\|^2 ,
\label{eq:ML4}
\end{equation}
where 
$\hat{\Pbf}_k$
is the orthogonal projection matrix onto the range space of $\hat{\Pibf}_k^{\perp}\hat{\abf}(\tau_k,\omega_k)\,$:
\begin{equation}
\hat{\Pbf}_k= 
\frac{\hat{\Pibf}_k^{\perp}\hat{\abf}(\tau_k,\omega_k) \hat{\abf}^H(\tau_k,\omega_k)\hat{\Pibf}_k^{\perp}}{\hat{\abf}^H(\tau_k,\omega_k)\hat{\Pibf}_k^{\perp}\hat{\abf}(\tau_k,\omega_k)} = \Ibf - 
\hat{\Pbf}^{\perp}_k\, .
\label{eq:Pdef}
\end{equation}
Now, the numerator of \eqref{eq:ML4} is precisely the 2D CAF used in the ECA algorithm, (see Equation (15) in \cite{colone_multistage_2009}). 
The inner product of the interference-cleaned data $\hat{\Pibf}_k^{\perp} \ybf_k$ with the steering vector $\hat{\abf}(\tau_k,\omega_k)$ is in effect a 2D matched-filtering operation to the target model at hypothesized parameters $\tau_k$ and $\omega_k$.
The denominator does not depend on the SC data, and can be absorbed into the threshold of a Constant False Alarm Rate (CFAR) detector. 
For most practical IO waveforms, the denominator is nearly independent of the target parameters when $|\omega_k| \gg 0$ or if $\tau_k$ is far from the range of clutter delays. However, 
for a slow moving target within the clutter range,
the full form of \eqref{eq:ML4} will improve the estimation performance over the 2D CAF without normalization.

\subsection{Target Localization}
\label{sec:TargLoc}

Given the transmitted information from all RNs, the task at the FC is to estimated the target location and velocity. 
As alluded to above, for high enough SNR, this can be done using local estimates at each RN, such as in classical source localization. See, e.g.,  \cite{MalanowskiMK:2012} for methods based on time-delays only and
\cite{HoX:2004,DuW:2014} for localization using time-delay and Doppler. However, such a two-step approach, albeit being more convenient from a communication point of view, is inherently suboptimal from a statistical performance perspective. Thus, we assume the available data to be (the sampled versions of) (\ref{eq:ML4}) from all nodes. Note, though, that our analysis can also be used to predict the performance of the two-step approach, since it provides the mean-square errors of the estimated delay and Doppler parameters at each node if desired.

Since the data are independent, the global likelihood function simply adds all contributions, scaled by their respective noise variances according to \eqref{eq:ML1}:
\begin{equation}
V(\thetabf) = \sum_{k=1}^K \frac{1}{\sigma_{e,k}^2} P_k(\tau_k(\thetabf),\omega_k(\thetabf))\, .
\label{eq:GlobalML}
\end{equation}
Here, $\thetabf = (x,y,v_x,v_y)^T$ represents the target position and velocity parameters, and
$\tau_k=\tau_k(\thetabf)$ and $\omega_k=\omega_k(\thetabf)$ are known functions of $\thetabf$,
as given by
(\ref{eq:tau})--(\ref{eq:Doppler}). 
For simplicity, we assume that all SC noise variances $\sigma_{e,k}^2 \overset{\Delta}{=} \sigma_e^2 $ are identical in what follows, and the weights are therefore omitted.
Likewise, we set the RC noise variances to also be the same at all RNs, denoted $\sigma_n^2$. These assumptions are natural if all RNs are equipped with identical hardware.

The global ML approach is now to maximize (\ref{eq:GlobalML}) w.r.t. the 4-dimensional target parameter vector $\thetabf$. For each hypothesized target localization and velocity, the corresponding time-delay and Doppler parameters are calculated for each node. The resulting value of (\ref{eq:ML4}) is added to the global likelihood function, and if $\thetabf$ does not correspond to a $(\tau,\omega)$ pair 
at a particular node, we simply add zero.
This is clearly a computationally very demanding task, and different search strategies are possible, such as coordinate-wise search and gradient-type optimization.
The practical implementation is, however, beyond the scope of the present paper. Our goal is to quantify the achievable performance with such an optimal approach, which also serves as a benchmark for suboptimal approximations.

At this point, we remark that \cite{colone_multistage_2009} also introduces the ECA-B version, where data is processed in batches. This reduces the computational cost and memory requirements at the RNs, since the interference cancellation is performed on data of a lower dimension. It also opens up the possibility to track time-varying scenarios,
see, e.g. \cite{Colone_etal_SlidingECA_TAES2016}, as well as to mitigate the effect of range migration. It is straightforward to extend the statistical analysis to cover the case of batch-wise data, see \cite{Viberg_etal_CAMSAP2025}. Due to space limitations we do not pursue this possibility any further herein.

\section{Performance Analysis}
\label{sec:perf}

The proposed global ML estimator is approximate in the sense that it uses the reference channel as if it were the true IO signal. It is of interest to quantify analytically the effect of this approximation. Specifically, what SNR is required in the reference channel in order for the approximate ML estimates to achieve the Cramér-Rao lower bound for the target parameters, assuming a perfect IO signal knowledge? 

\subsection{Consistency}

The first step is to establish consistency, in the sense that for noise-free data, the criterion function (\ref{eq:ML4}) is maximized by the true target delay-Doppler pair, denoted $\tau_0$ and $\omega_0$ respectively, as seen from the $k$-th RN. It is easy to establish this result if the IO signal is such that the steering vector is \textit{unambiguous}. By this we mean that the following holds true over the range of target parameters of interest:
\begin{equation}
\text{rank}\Bigl[
\abf(\tau_0,\omega_0)\ \abf(\tau,\omega)
\Bigr] = 1\ \Longleftrightarrow \ (\tau_0,\omega_0)=(\tau,\omega)\,,
\label{eq:unambiguous}
\end{equation}
for all \( \tau, \tau_0 \in \mathbb{R} \) and \( \omega, \omega_0 \in (-\pi / \Delta T, \pi / \Delta T] \). As before, the steering vector is defined by $\abf(\tau,\omega) = \sbf(\tau) \odot \vbf(\omega)$. That the condition in \eqref{eq:unambiguous}
implies
statistical consistency of the estimator can be proved as follows: when the noise variances $\sigma_n^2\rightarrow 0$ and $\sigma_e^2\rightarrow 0$ vanish simultaneously, the criterion (\ref{eq:ML2}) converges in the mean-square sense to its noise-free version
\begin{equation} 
\ell _0(\etabf) =
 \| \Sbf_I \fbf_0 +
    d_{0}\, \abf(\tau_{0},\omega_0)
     - \Sbf_I \fbf -
    d\, \abf(\tau,\omega) \|^2,
\label{eq:noisefreecrit}
\end{equation}
where $\etabf=\{\fbf,d,\tau,\omega\}$ represents the set of unknown parameters, $\etabf_0$ its ``true" value, and where we have replaced $\xbf(\tau)$ and $\Xbf_I$ by the noise-free versions $\sbf(\tau)$ and $\Sbf_I$, since the amplitude is irrelevant here. Provided the convergence is uniform in $\etabf$, which is true if $\sbf(\tau)$ has bounded derivatives, then $\hat{\etabf}$ converges in the mean-square sense to the minimizer of (\ref{eq:noisefreecrit}). 
Clearly, $\ell _0(\etabf)\geq 0$, with equality if and only if
\begin{equation}
\Bigl[
\Sbf_I\ \ \abf(\tau_{0},\omega_0)\ \ \abf(\tau,\omega)
\Bigr] \left[
\begin{array}{c}
     \fbf_0 - \fbf \\ d_0 \\ -d\ \       
\end{array}
\right] = 0\, .
\label{eq:uniqueness}
\end{equation}
Provided the target is present so that $|d_0|>0$,
the above is possible only if $\tau=\tau_0$ and $\omega=\omega_0$, if (\ref{eq:unambiguous}) holds 
and if the vector $\abf(\tau_0,\omega_0)$ is not in the span of $\Sbf_I$. It also follows from (\ref{eq:unambiguous}) that the latter is satisfied whenever the target has non-zero Doppler, $\omega_0 \neq 0$.
In principle, it is also possible to localize a stationary target, provided it is known to reside outside the clutter-contaminated range bins. But since we are more interested in the case where the target is obscured by clutter, we consider only the case $\omega_0 \neq 0$.

The following key
result clarifies that \eqref{eq:unambiguous} is satisfied for a certain class of stochastic processes. The proof is deferred to Appendix \ref{Appendix3}.

\begin{theorem} \label{thm:rf_unambiguous}
Let \( s(t) \) be a circularly symmetric complex-valued bandlimited wide-sense stationary Gaussian stochastic process, with zero mean and spectral density \( \mathcal{S}_s (f) = 1 \) supported in \( [-B/2,B/2]\), \(B > 0\). Let \(0 \le t_1 < \dots < t_N\), \(N \ge 3\), be fixed real sampling times, and \( \tau_0, \omega_0 \) fixed real parameters. With the steering vector \( \abf(\tau,\omega) = \sbf(\tau)\odot\vbf(\omega) \) of \eqref{eq:adef}, i.e.\ \( [\abf(\tau,\omega)]_n = s(t_n - \tau)e^{j\omega t_n} \), the following holds: excluding the finitely many delays \( \tau \in \{\tau_0\}\cup\{\tau_0 + t_p - t_q : p\neq q\} \), almost surely there is no \( (\tau,\omega) \), with \( \tau \) ranging over any compact set of admissible delays and \( \omega\in\mathbb{R} \), for which \( \abf(\tau,\omega) \) and \( \abf(\tau_0,\omega_0) \) are linearly dependent.
\end{theorem}

In what follows, we will assume (\ref{eq:unambiguous}) to be true so that the time-delay and Doppler parameters can be uniquely determined at each RN. 
It is also assumed that the geometry of the scenario is such that $\thetabf$ can be uniquely retrieved from the $\{(\tau_k,\omega_k)\}_{k=1}^K$ pairs.
This is obviously true for a ``large enough'' number of RNs, placed, say, in a random fashion. A precise statement regarding the requirement
appears to be an open problem, but a complete solution was recently provided for the related GPS problem \cite{BOUTIN_GPSunique_2024}.

\subsection{Asymptotic Covariance Matrix}

The global ML estimate of the target parameters is obtained by maximizing (\ref{eq:GlobalML}) w.r.t. $\thetabf$. The analysis can also be applied to a single RN, in which case $V(\thetabf)=P_k(\tau,\omega)$ and $\thetabf=(\tau,\omega)^T$. For either case, 
treating $\thetabf$ as a continuous-valued parameter, the gradient is zero at the optimal value,
\begin{equation}
V'(\hat{\thetabf}) = 0\, ,
\label{eq:gradient}
\end{equation}
where $\hat{\thetabf}$ denotes the estimate. Assuming a large number of samples or high SNR, $\hat{\thetabf}$ will be close to the true value, say $\thetabf_0$. Then, a standard first-order expansion of (\ref{eq:gradient}) gives the first-order expression for the estimation error
\begin{equation}
\hat{\thetabf} - \thetabf_0 \simeq -\Hbf^{-1}V'(\thetabf_0)\, ,
\end{equation}
where $\Hbf$ is the limiting Hessian matrix as 
$\sigma_n^2,\sigma_e^2\rightarrow 0$. From this expression, we can compute the asymptotic covariance matrix of the estimation error as
\begin{equation}
\E [(\hat{\thetabf} - \thetabf_0) (\hat{\thetabf} - \thetabf_0)^T ] \approx
\Hbf^{-1}\E[V'(\thetabf_0) 
V'^T(\thetabf_0)]\,\Hbf^{-1} .
\label{eq:AsCovdef}
\end{equation}
We consider only the case of ``sufficiently high" SNR, i.e. $\sigma_n^2,\sigma_e^2$ being small enough. 
As demonstrated in the next section using Monte-Carlo simulations, this means that the estimation error is small, so that higher-order terms can be neglected.
For large $N$, this can be the case for target SNR values well below 0 dB.
In the following, we provide compact expressions for the Hessian matrix and the
first-order approximation (in $\sigma_n^2$ and $\sigma_e^2$) of the covariance matrix of the gradient appearing in (\ref{eq:AsCovdef}).

In order to express the approximate covariance matrix in a compact form, the following notation is first introduced. Let
\begin{equation}
\Dbf_k = \left[ \frac{\partial \abf(\tau_k(\thetabf),\omega_k(\thetabf))} {\partial\theta_1},\dots,\frac{\partial \abf(\tau_k(\thetabf),\omega_k(\thetabf))} {\partial\theta_4}
\right]
\label{eq:Ddef}
\end{equation}
denote the Jacobian matrix of the $k$-th steering vector, where $\thetabf=(\theta_1, \dots , \theta_4)^T = (x,y,v_x,v_y)^T$. 
For the single-node case, $\thetabf=(\tau,\omega)^T$ and 
$\Dbf_k$ has only two columns.
Further, define the $N\times N(L+1)$ matrix
\begin{equation}
\Zbf_k = \left[
b_k\,\Ibf + d_k\,\text{diag}(\vbf(\omega_k))\quad \cbf_k^T\otimes \Ibf
\right]\, ,
\label{eq:Zdef}
\end{equation}
and let $\Jbf_k$ be a selection matrix such that
$$
\text{vec}(\Sbf_I) = \Jbf_k \sbf_I\,,
$$
where $\sbf_I=[
s_k(t_{-L}),\dots,s_k(t_{N-1})
]^T$. 
The three components in (\ref{eq:Zdef}) model the error contributions to $\hat{\thetabf}$ due to not knowing the IO waveform perfectly (see \eqref{eq:T2noise} in Appendix \ref{Appendix1}). The first term to the left is due to imperfect DPI cancellation and the second handles the effect of using the incorrect steering vector in (\ref{eq:ML4}), i.e. a ``mismatched filter". The right component comes from the imperfect clutter cancellation. 

Further, let $\Pbf_k$ be the noise-free version of $\hat{\Pbf}_k$, defined in (\ref{eq:Pdef}), and
introduce the orthogonal projection matrix
\begin{equation}
\Tilde{\Pbf}_k = \Pibf_k^{\perp} \Pbf_k^{\perp}\Pibf_k^{\perp},
\label{eq:Ptilde}
\end{equation}
which projects onto $\text{span}(\Pibf_k^\perp\Bbf)$, where $\Bbf$ is any matrix that spans the orthogonal complement of the steering vector $\abf$. We can now state the second main contribution of this paper.
\begin{theorem}
Let $\hat{\thetabf}$ be obtained by maximizing (\ref{eq:GlobalML}) and assume the IO waveform be such that (\ref{eq:unambiguous}) holds. Assume further that the radar scenario is such that the 
$\{(\tau_k,\omega_k)\}_{k=1}^K$ pairs together uniquely determine $\thetabf$. Then, as $\sigma_n^2\rightarrow 0$ and $\sigma_e^2\rightarrow 0$ jointly, we have $\hat{\thetabf} \rightarrow\thetabf_0$ in probability; and its covariance matrix is to first order given by 
\begin{equation}
\E [(\hat{\thetabf} - \thetabf_0) (\hat{\thetabf} - \thetabf_0)^T ] = \Cbf_{\boldsymbol{\theta}} =  \textbf{CRB}_{\boldsymbol{\theta}} +
\Hbf^{-1} \Qbf \,\Hbf^{-1} ,
\label{eq:CRBascov} 
\end{equation}
where $\textbf{CRB}_{\boldsymbol{\theta}}  = -\sigma_e^2\,\Hbf^{-1}$ and
\begin{align}
\Hbf &= 
- 2 \sum\limits_{k=1}^K |d_k|^2\, \Re \left\{ \Dbf^H _k\Tilde{\Pbf}_k\,\Dbf_k^{}\right\} 
\label{eq:CRB_theta}
\\
\Qbf &= 2\sigma_n^2\,\,\sum\limits_{k=1}^K \frac{|d_k|^2}{|a_k|^2}\, \Re \left\{ \Dbf^H_k \Tilde{\Pbf}_k\,\Zbf_k^{}\Jbf_k^{} \Jbf_k^T\Zbf_k^{H}\Tilde{\Pbf}_k\,
\Dbf_k^{}\right\} \,.
\label{eq:Q}
\end{align}
\end{theorem}

\begin{proof} The covariance matrix of the gradient and the asymptotic Hessian matrix are derived, respectively, in Appendices \ref{Appendix1} and \ref{Appendix2}. 
Using (\ref{eq:AsCovdef}) then leads to (\ref{eq:CRBascov})--(\ref{eq:Q}). \end{proof}

\begin{remark}
The term $\textbf{\textit{CRB}}_{\boldsymbol{\theta}}$ in (\ref{eq:CRBascov}) is the CRLB for $\thetabf$ assuming a perfectly known IO signal, and the second term quantifies the excess error due to the noise in the reference channel. That (\ref{eq:CRB_theta}) is the CRLB for $\thetabf$ assuming a perfect reference signal follows from the fact that when $\sigma_n^2=0$, $\hat{\thetabf}$ is the exact ML estimate of $\thetabf$, which is statistically efficient to first order under mild conditions \cite{Kay:93}. 
\end{remark}

\subsection{Statistical Efficiency}

Using the results from the previous section, we can now address the question of how large the errors in the reference signal can be in order for its effect to be negligible. We remind that while the errors are assumed to emanate from a noisy reference channel, the analysis is useful also if the reference waveform is generated in another way, or if the errors are
due to other effects, such as phase noise. The important assumption is that the reference signal error is additive and white, not necessarily Gaussian distributed, and that its variance is small. 

The expression (\ref{eq:CRBascov}) has two terms, corresponding to the contributions to the estimation error from the noise in the surveillance channel (first term) and the reference channel (second term), respectively.
Thus, the question is how big the second term in (\ref{eq:CRBascov}) is compared to the CRLB term. Upon comparing (\ref{eq:CRB_theta})--(\ref{eq:Q}), we can conclude that the matrix $\Zbf_k^{}\Jbf_k^{}\Jbf_k^{T}\Zbf_k^H$ plays a crucial role. Specifically, 
the second term is negligible if it holds that
$$
\left\| \Zbf_k \Jbf_k\right\|_2^2 \ll \frac{\sigma_e^2}{\sigma_n^2}\,|a_k|^2\, .
$$

We can bound the left hand side using (\ref{eq:Zdef}) as
\[
\begin{split}
\left\| \Zbf_k\Jbf_k \right\|_2^2 & \leq 
\left\| \Jbf_k \right\|_2^2 
\left\| \Zbf_k \right\|_2^2 \leq (L+1)  \\ & 
\cdot \left(
\left\| (\cbf_k^T\otimes \Ibf) \right\|_2^2
+ \left\| b_k\Ibf +
d_k\,\text{diag}(\vbf(\omega_k)) \right\|_2^2 
\right)\, ,
\end{split}
\]
where we have used that $\left\| \Jbf_k \right\|_2^2 = L+1$, which is easily verified,
since $\Jbf_k$ is a selection matrix with orthogonal columns and the maximum diagonal element of $\Jbf_k^T\Jbf_k^{}$is $L+1$. 
Thus, we conclude that the proposed algorithm will achieve the CRLB using a perfectly known reference signal if the following condition is met for all $k$:
\begin{equation}
(L+1)\,\frac{|b_k|^2 + |d_k|^2 + \|\cbf_k\|^2}{\sigma_e^2} \ll
\frac{|a_k|^2}{\sigma_n^2}\, .
\label{eq:interferencebound}
\end{equation}
The right-hand side of (\ref{eq:interferencebound}) is the SNR in the RC, and the left-hand side is an upper bound on the total interference-to-noise ratio in the SC. 
Thus, we conclude that the requirement on the SNR in the RC to be able to ignore the RC noise depends on the powers of the DPI, the target and the clutter, and they are all scaled by the dimension of the interference subspace.

\begin{figure*}[!t]
  \centering
  \subfloat[]{%
    \includegraphics[width=0.49\textwidth]{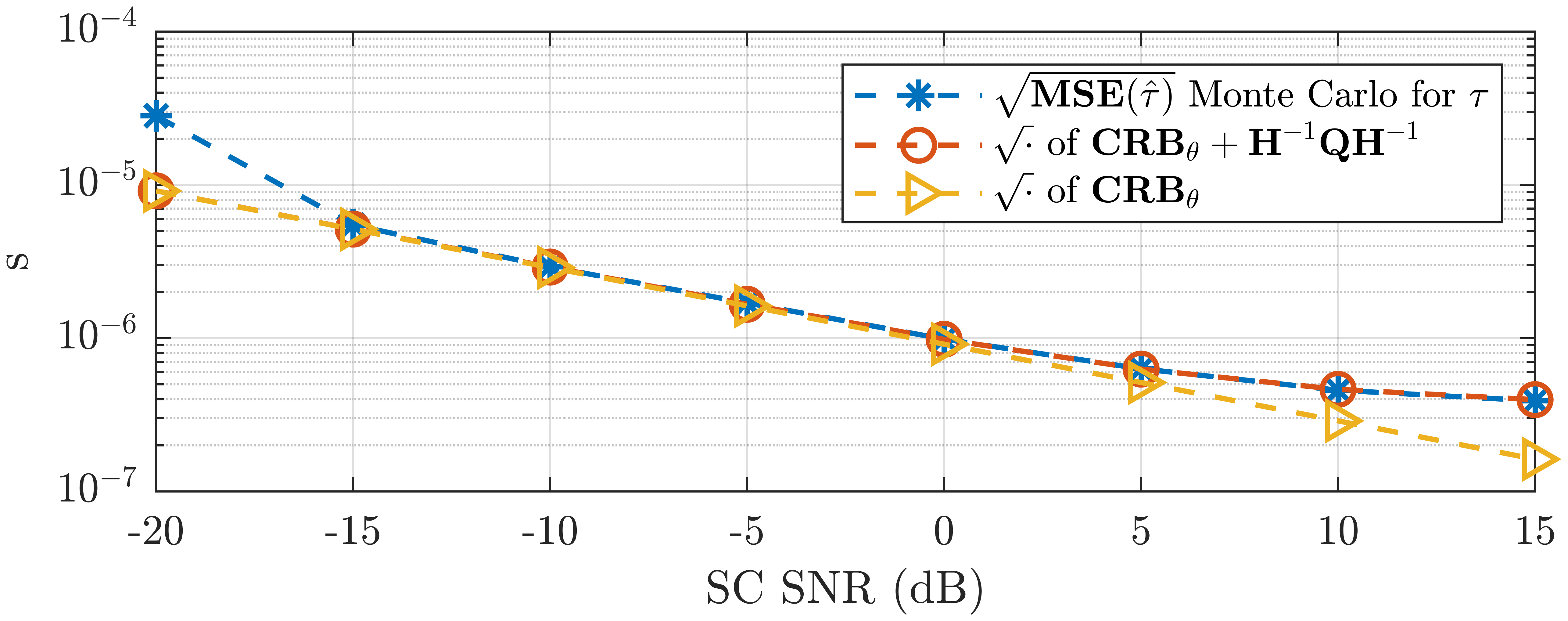}%
    \label{fig:tauSCsweep}}%
  \hspace{0.01\textwidth}%
  \subfloat[]{%
    \includegraphics[width=0.49\textwidth]{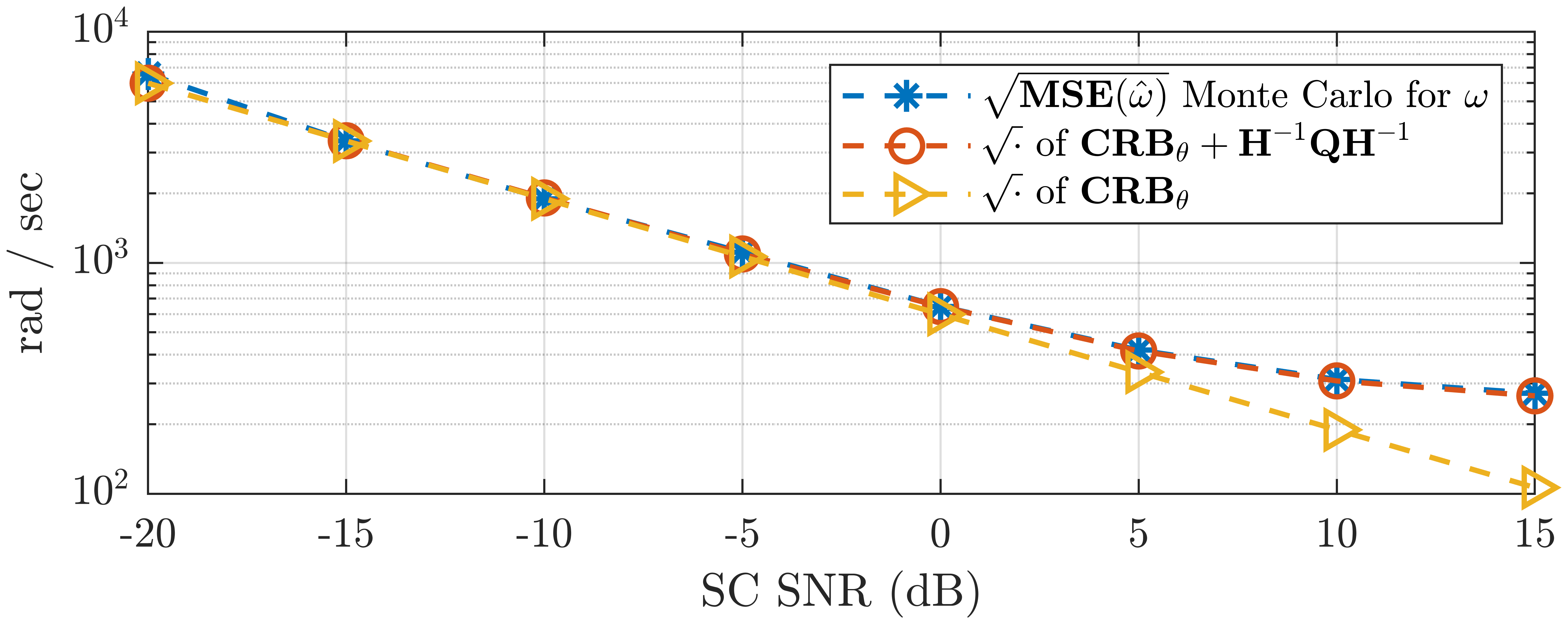}%
    \label{fig:omegaSCsweep}}
    \\
    \subfloat[]{%
    \includegraphics[width=0.475\textwidth]{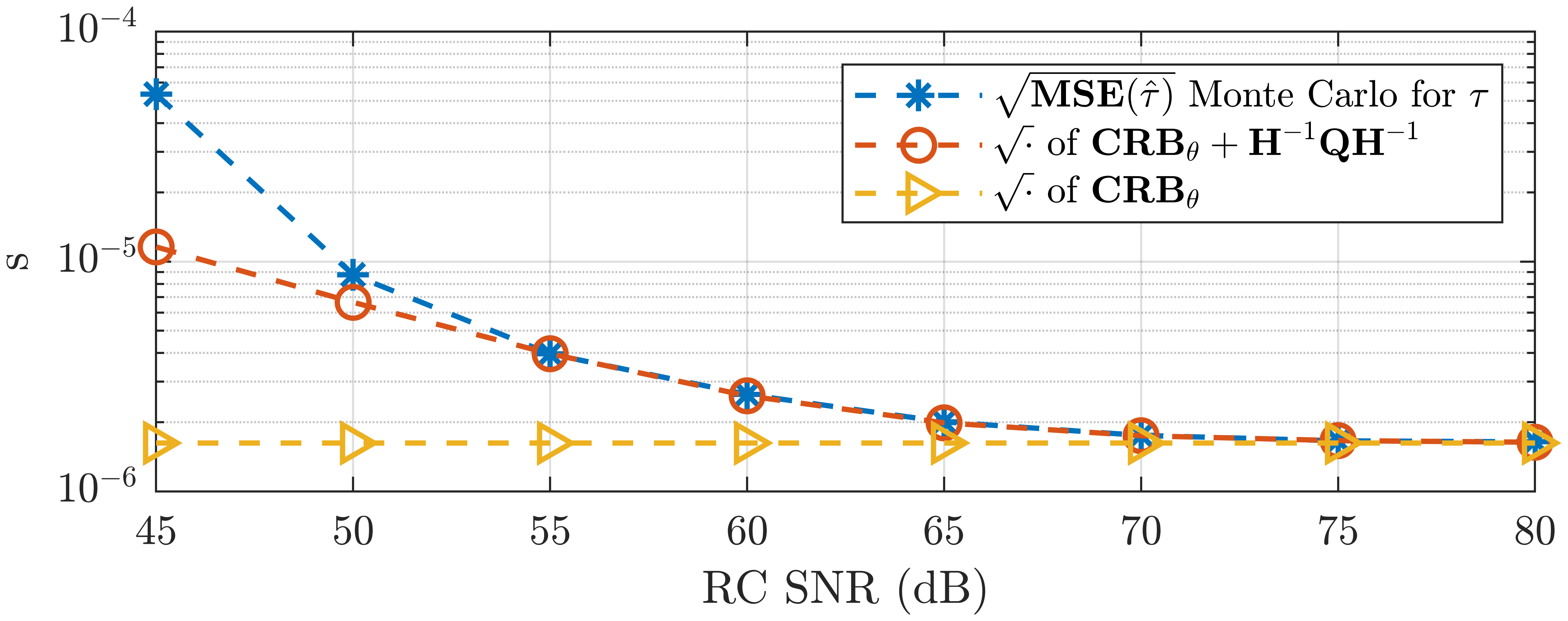}%
    \label{fig:tauRCsweep}}%
  \hspace{0.01\textwidth}%
  \subfloat[]{%
    \includegraphics[width=0.49\textwidth]{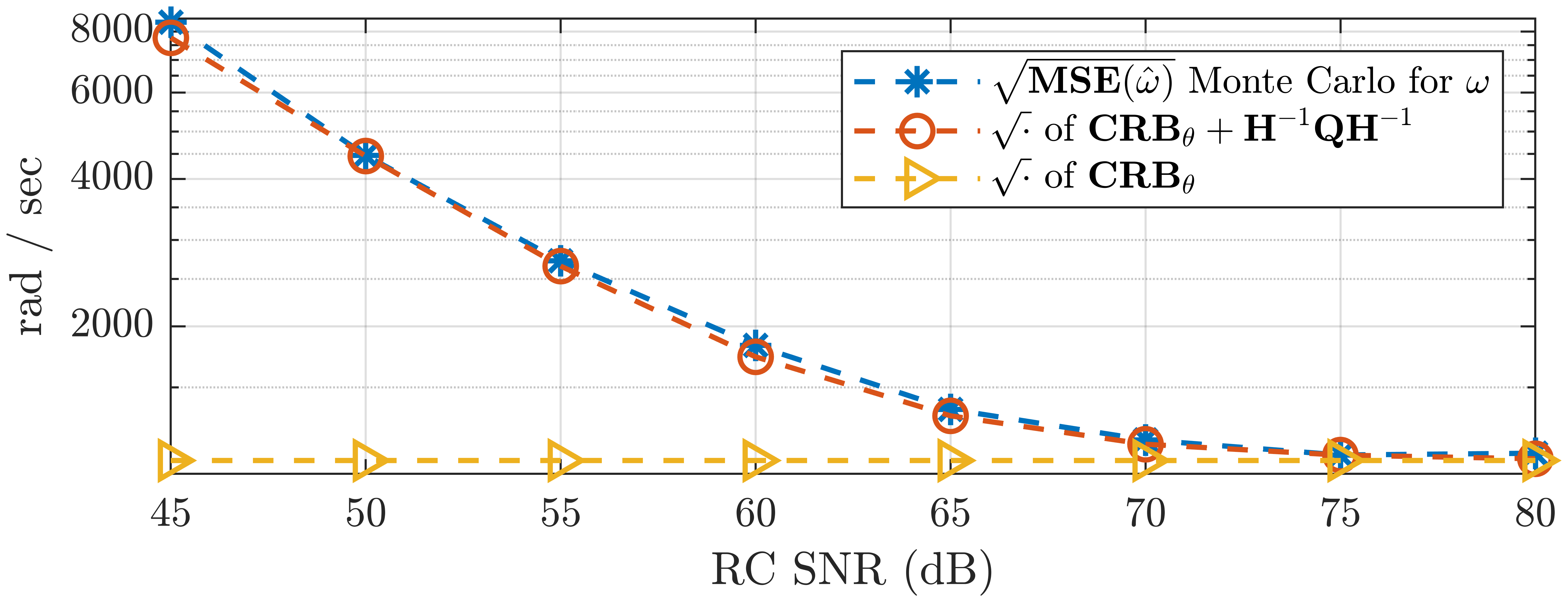}%
    \label{fig:omegaRCsweep}}
    \caption{Empirical and theoretical MSE for estimating time-delay (left plots) and Doppler frequency (right plots) versus SC SNR (top) and RC SNR (bottom) respectively. The RC SNR was kept constant to \( 75 \text{ dB} \) in \eqref{fig:tauSCsweep}-\eqref{fig:omegaSCsweep} while the SC SNR was kept constant to \( -5 \text{ dB} \) in \eqref{fig:tauRCsweep}-\eqref{fig:omegaRCsweep}.}
  \label{fig:base_bist_analysis}
\end{figure*}

\section{Numerical Examples}
In this section we test and validate the theory developed so far with numerical simulations. In all the following subsections the baseband IO sensing waveform is chosen to be a band\-limited, stationary Gaussian process with \(16\,\text{MHz}\) bandwidth and unit variance. The waveform is subsequently up- and down-converted to and from Radio Frequency (RF) using a carrier of \(600\,\text{MHz}\). Each RN samples the incoming waveform at \(25\,\text{MSps}\). These design values are chosen to closely reflect the real characteristics of the waveform used at the ``Brudaremossen'' TV tower station, located on the outskirts of Göteborg, Sweden. The wave is assumed to propagate along straight lines, and the amplitudes \(a_1\), \(b_1\) and \(d_1\) are determined using the (deterministic) bistatic radar equation. All ``empirical'' numerical simulations include a clutter area with a clutter filter of length \(L=70\), whose taps are generated at random at each RN. The clutter area is assumed to be made of strongly reflecting elements so that the clutter-to-noise ratio is \( \sim100 \) dB on average in Sections \ref{subsec:base_nexp} and \ref{subsec:tracking}, \( \sim60 \) dB in Section \ref{subsec:complete_mstatic}, whereas it is varied in \ref{subsec:theoretical_insights}.

\subsection{Baseline validation: single bi-static pair}
\label{subsec:base_nexp}
The first set of numerical experiments consists of Monte Carlo simulations to validate the formula \eqref{eq:CRBascov} in a bi-static setup, expanded from \cite{Viberg_etal_CAMSAP2025}. 
Thus, we set \( K = 1 \) and use \( \thetabf = (\tau, \omega)^T  \) as target parameters. 
A single target is present, with radar cross-section \( \approx 0.02 \text{ m}^2 \), and moving at constant speed of \( \approx 200 \text{ m/s} \) (with constant velocity vector \( (70, -190) \text{ m/s} \) and corresponding to \(\approx 1320 \text{ rad/sec } \) in our bistatic setup). The acquisition time equals \( N = 2^{13} \) samples, corresponding to approximately \( 0.33 \text{ ms} \) of data at the sampling rate $F_s = 1/\Delta T = 25$ MSps.

In Figure \ref{fig:base_bist_analysis}, we present the results for two scenarios using different SNR levels, defined as $20\log_{10}(|d_k|/\sigma_e)$ (SC) and $20\log_{10}(|a_k|/\sigma_n)$ (RC), respectively. The empirical RMSE (Root-Mean-Square-Error) values are based on \(6000\) noise instances, and in each the criterion function \eqref{eq:GlobalML} is maximized using the Nelder-Mead algorithm with the ``true'' parameter values \( \thetabf_0 \) as starting point. This initialization is used, since we are concerned with the quality of the global optimum rather than the performance of a specific implementation.
Figure \ref{fig:base_bist_analysis} summarizes the numerical simulations as described here. It can be seen that the theory agrees well with the empirical results above a certain threshold, which in this case is \(-15\) dB in the SC (Figure \ref{fig:tauSCsweep}) or 55 dB in the RC (Figure \ref{fig:tauRCsweep}) for the time-delay estimate. This illustrates that the first-order approximations are valid whenever the estimation error is ``small enough'', which indeed can occur at quite small SNR values. 
We may claim that the theory predicts the empirical performance well under normal operating conditions also in a passive radar setting with low target SNR.

Figure \ref{fig:base_bist_analysis} also visualizes the excess error caused by noise in the RC. The mismatch between the MSE of the estimator and the CRLB is correctly ``predicted'' by the covariance \eqref{eq:CRBascov}, and it occurs both at a low SNR in the RC and at a high SNR in the SC, as seen in \eqref{eq:interferencebound}.
In this case, the inequality in \eqref{eq:interferencebound} is in fact violated at the critical point, underscoring that this condition is sufficient but not necessary.

\begin{figure}[htbp]

  \centering
  \begin{subfigure}[b]{0.47\textwidth}
    \centering
    \includegraphics[width=\linewidth]{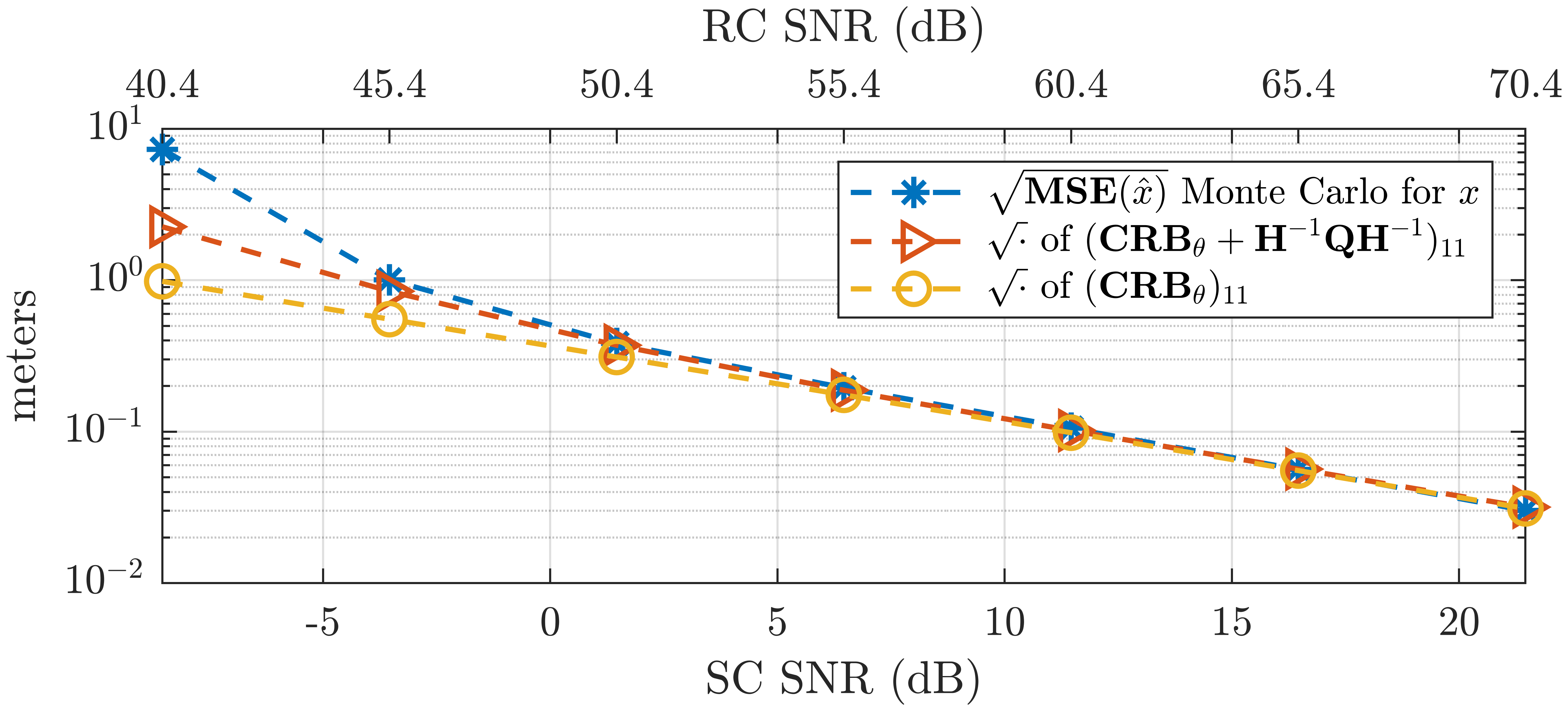}
    \caption{}
     \label{fig:tau_bllensweep_multi}
  \end{subfigure}
  
  \begin{subfigure}[b]{0.47\textwidth}
    \centering
    \includegraphics[width=\linewidth]{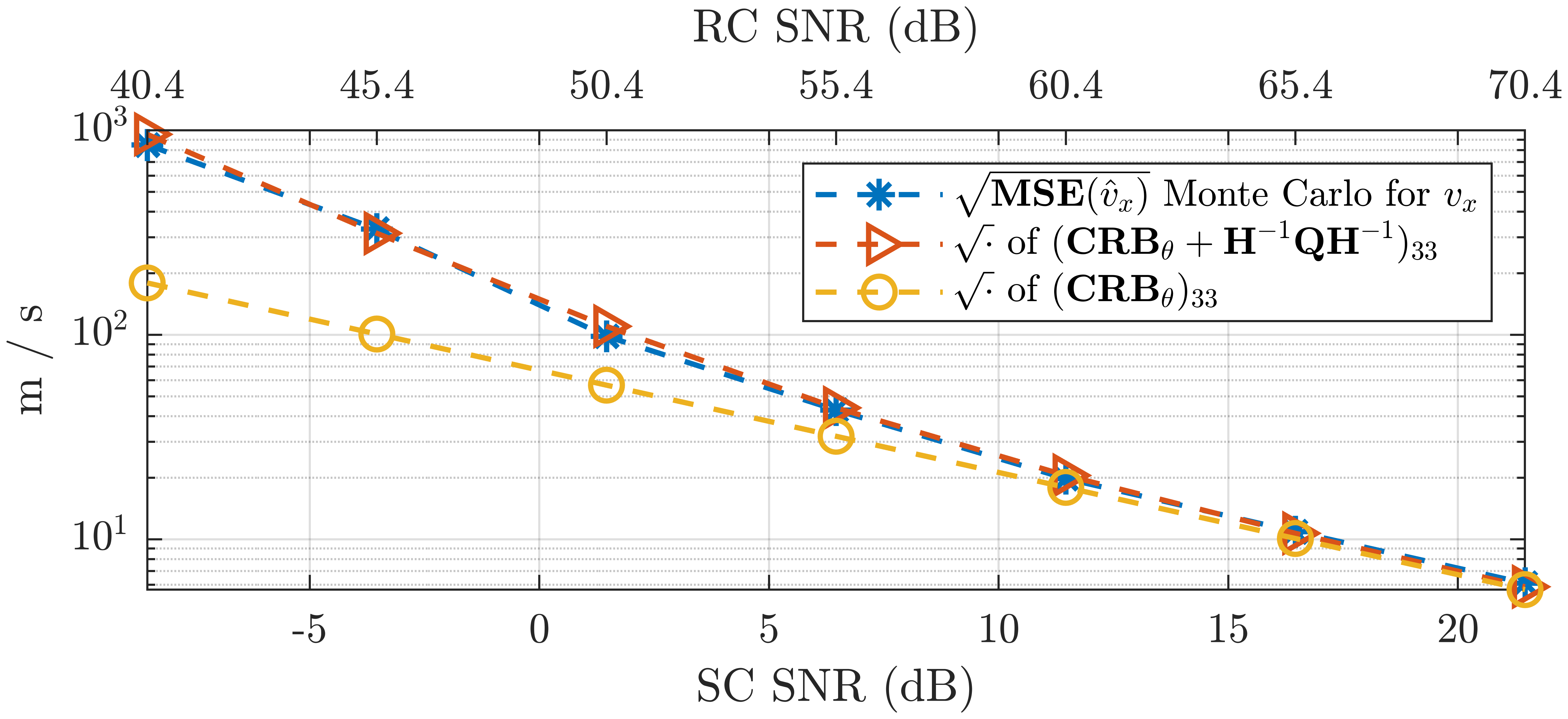}
    \caption{}
    \label{fig:omega_bllensweep_multi}
  \end{subfigure}

  \caption{Monte Carlo simulation versus theory for localization with \(4\) RNs. Plots for \(y\) and \(v_y \) are qualitatively identical and thus omitted. SC SNR is an average over the different RNs.}
  \label{fig:both_bllensweep_multi}
\end{figure}

\subsection{Target localization: complete multi-static setup}
\label{subsec:complete_mstatic}

The next step is to validate the formulas \eqref{eq:CRB_theta} and \eqref{eq:Q} in a multi-node setup. The main difference with the single bi-static setup lies in an ambiguity reduction, that makes it meaningful to estimate the full suite of parameters \( \thetabf = (x,y,v_x,v_y)^T \). We remind 
that in the bistatic geometry, the relations between \( (\tau, \omega) \) and \( (x,y,v_x,v_y) \) are given by \eqref{eq:tau} and \eqref{eq:taudot}. With this in mind, we perform target localization and velocity estimation using a similar radar scene as in the previous subsections, but with \(K=4\) RNs, located at \( (-300,300), \ (300, 300), \ (300, -300) \text{ and } (-300,-300) \) meters respectively, while the IO sits at \( (0,0) \), and the target moves with motion equation \( (5,-10) + (100\sqrt{2}, 100 \sqrt{2})t  \), \( t \ge 0 \). Since we assumed constant (thermal) noise power at each RN, the SNR was changed by increasing the transmitter output. This has the effect of changing the SNR in both the SC and the RC, which is why the MSE is approaching the CRLB for increasing SNR. 
Figure \ref{fig:omega_bllensweep_multi} again shows excellent agreement between the empirical and theoretical RMS values. The same holds for Figure \ref{fig:tau_bllensweep_multi}, although it can be seen that at SC SNR \(\leq -4\) dB, yielding an RMS location error of more than 1 m, the 
SNR is not high enough for the first-order covariance expansion \eqref{eq:CRBascov} to fully describe the excess error. It is also confirmed once more that a very high RC SNR is required to reach the CRLB that assumes a perfect reference signal. 

\subsection{Theoretical SNR insights}
\label{subsec:theoretical_insights}

\begin{figure}[htbp]
    \includegraphics[width=0.95\linewidth]{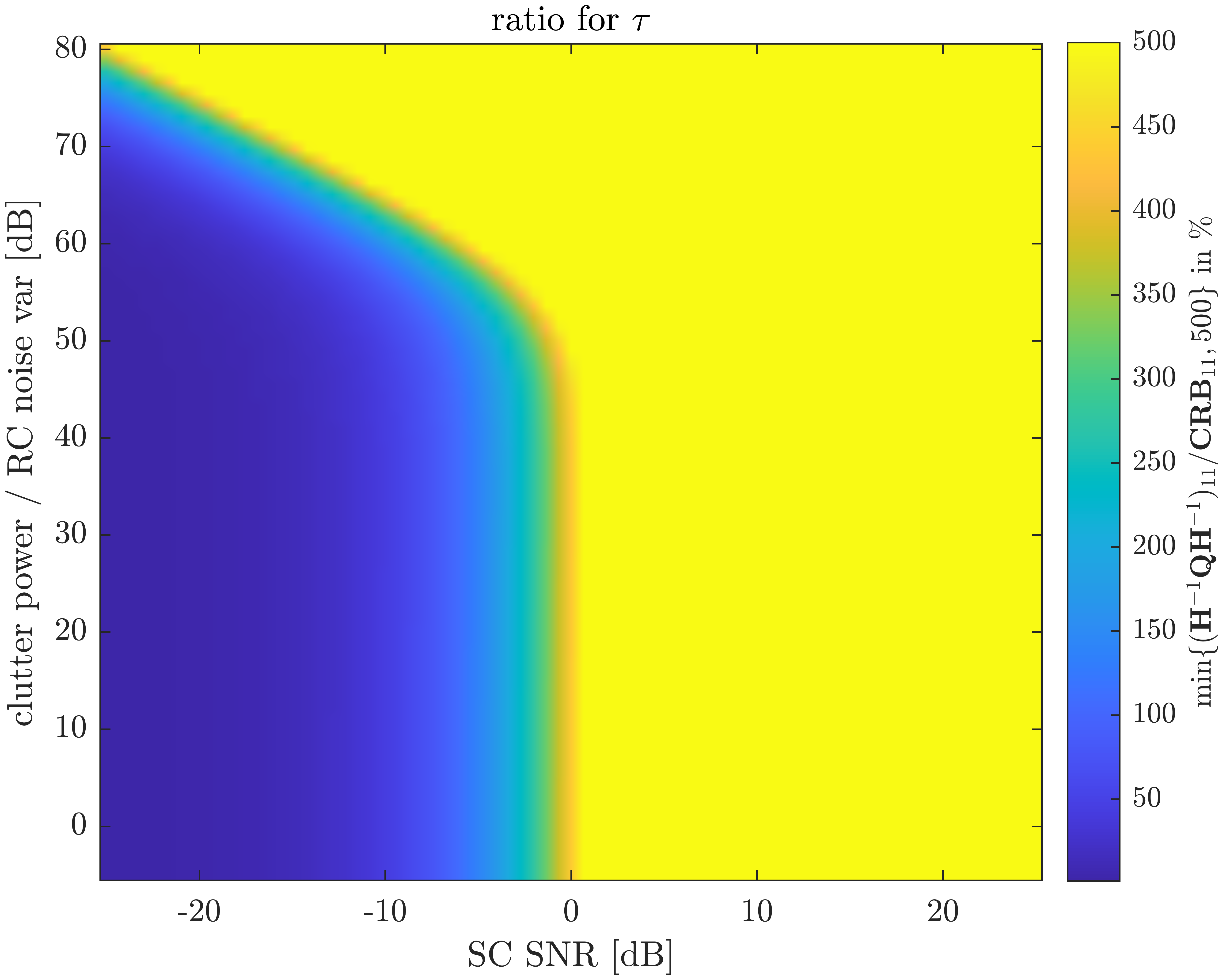}
  \caption{Ratio in \( \%\) between the $\tau$ components of the first order term $\Hbf^{-1}\Qbf\Hbf^{-1}$ in \eqref{eq:CRBascov} and the CRLB term $\textbf{\textit{CRB}}_{\boldsymbol{\theta}} = -\sigma_e ^2 \Hbf^{-1}$, with clipping at \( 500 \%\). The plot for \(\omega\) is similar and thus omitted.}
  \label{fig:theor_snr_insights}
\end{figure}
Figure~\ref{fig:theor_snr_insights} displays the interplay between \eqref{eq:CRB_theta} and \eqref{eq:Q} as functions of the SC SNR and the clutter-to-RC noise ratio. The figure shows, in percentage, the ratio
\((\Hbf^{-1}\Qbf\Hbf^{-1})_{11} / (\textbf{\textit{CRB}}_{\boldsymbol{\theta}})_{11} \), which equals the ``excess error'' \( (\Cbf_{\boldsymbol{\theta}})_{11}/(\textbf{\textit{CRB}}_{\boldsymbol{\theta}})_{11} - 1  \), where $\Cbf_{\boldsymbol{\theta}}$ is the total covariance matrix \eqref{eq:CRBascov}. The bistatic radar scene with one receiver and one transmitter was held fixed, while \( \sigma_e^2 \) and the clutter power were varied to achieve the desired SNR values. The RC SNR was kept constant at \(60\,\text{dB}\). The contour plot reveals two distinct regimes separated by a narrow transition region: the blue zone in which the CRLB dominates the excess error due to noise in the reference channel, and the yellow in which the opposite holds. For low or moderate clutter, the transition occurs at a constant SC SNR value around 0 dB. As the clutter power increases beyond a critical value, the transition occurs at an increasingly lower SNR in the SC.

\subsection{Target tracking: complete multi-static setup}
\label{subsec:tracking}
\begin{figure}[htbp]
  \centering
    \includegraphics[width=0.45\textwidth]{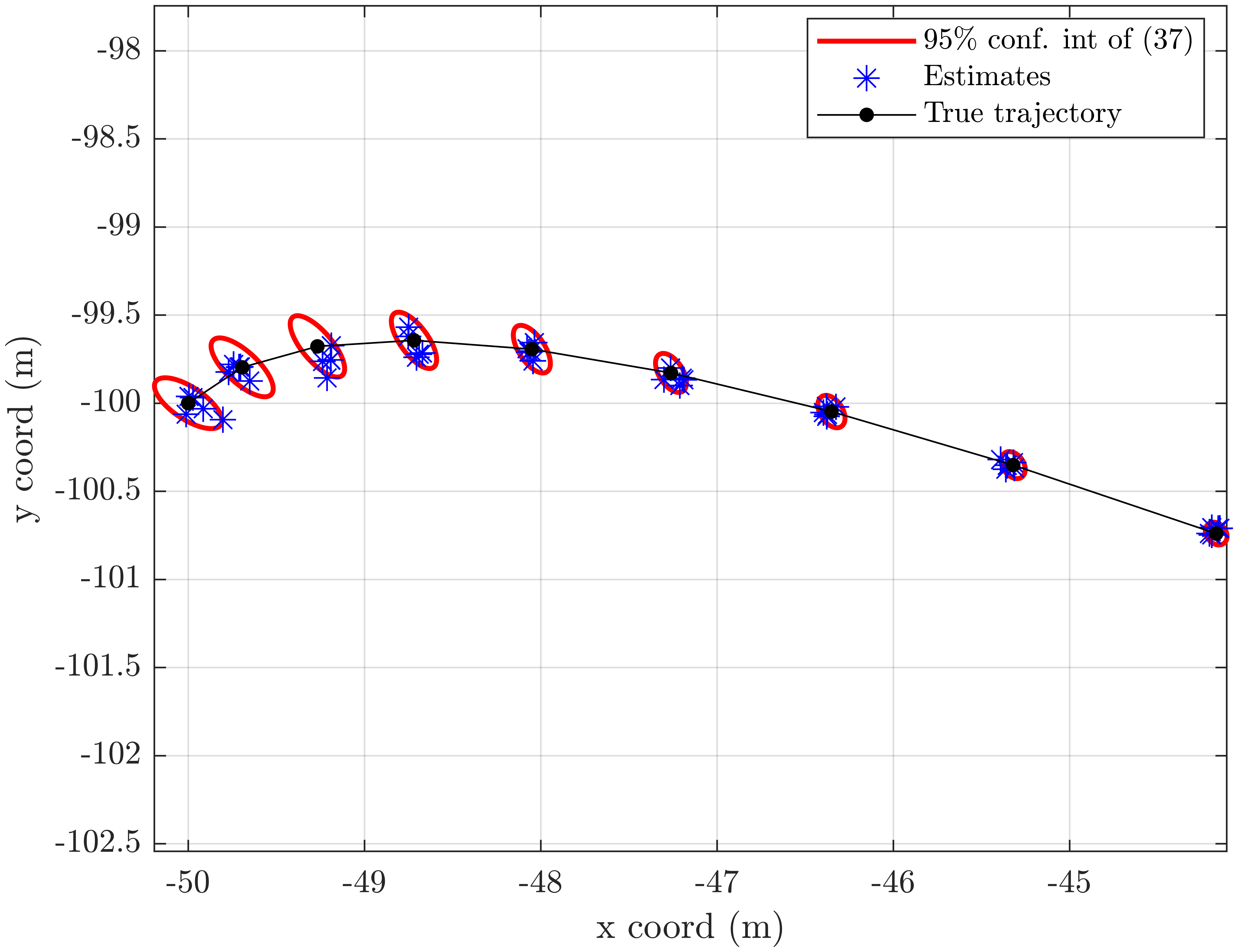}\\[1ex]
    \includegraphics[width=0.45\textwidth]{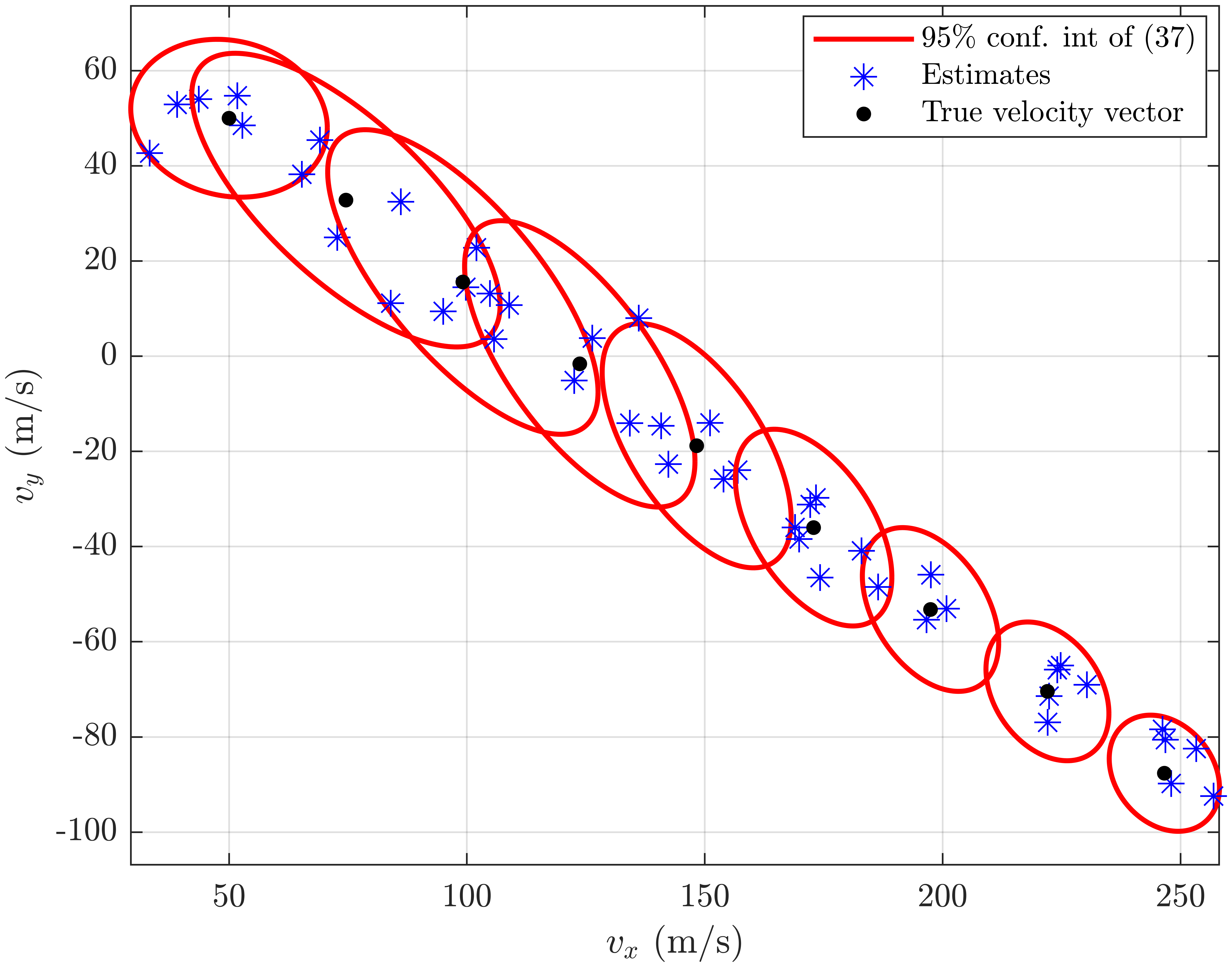}%
    \caption{Target tracking illustration.}
  \label{fig:trackers_posandvel}
\end{figure}
To conclude the numerical simulations, we illustrate in Figure \ref{fig:trackers_posandvel} an example of a ``target tracking scenario", where position and velocity estimation is performed within a multi-static setup. No post-filtering is done on the estimates, though. In this scenario, we have \(5\) RNs located at points \( (300 \cos(\alpha_k), 300 \sin(\alpha_k)) \) meters with \( \alpha_k = 2 \pi (k -1)/5 \), \( k=1,\dots,5 \). The IO is still located at the origin. We run 5 Monte Carlo noise realizations per sensing interval, and plot them altogether along with \( 95 \% \) confidence ellipses using \eqref{eq:CRBascov}. The size of the ellipses change as a result of the scenario geometry, as well as the target's acceleration, which is 
non-zero in this example.

\section{Conclusions}

The main contribution of the present paper is a statistical characterization of the
Extensive Cancellation Algorithm (ECA) \cite{colone_multistage_2009}, 
which performs
target localization using spatially separated passive radar receivers. 
Each receiver is equipped with a Reference Channel (RC), with a line-of-sight connection with the illuminator of opportunity, and a Surveillance Channel (SC) that captures reflections from the target of interest. 
The ECA algorithm uses 
data from the RC to cancel interference in the SC, in the form of direct-path interference and clutter from stationary objects. After canceling the interference, the time-delay and Doppler parameters are estimated using a non-linear least-squares approach that results in a 2-D spectrum (or normalized cross-ambiguity function). A Fusion Center (FC) collects information from all Receiving Nodes (RNs), and calculates the final target position and velocity estimates. Our analysis is based on the so-called Direct Position Determination (DPD) approach, in which the complete spectra from all RNs are assumed available at the FC. However, the results are also useful for the two-step method, where each node calculates the distance and radial velocity of the target, and only these parameters with corresponding MSE estimates are transmitted to the FC.

Our analysis provides sufficient conditions for the target parameter estimates to be consistent in the noiseless case, and this is shown to hold 
for a general class of
transmitted signals. Further, we derive the covariance matrix of the target parameter estimates under the assumption of sufficiently small noise variances in both the reference and the surveillance channel. This is expressed as the sum of the Cramér-Rao lower bound for a noise-free reference signal and a penalty term that explicitly captures the effect of noise in the reference channel. Based on this, we derive an explicit sufficient condition for the effect of the RC noise to be negligible compared to the effect of noise and interference in the SC.

The theoretical MSE expressions are found to agree well with the results of Monte-Carlo simulations in a wide range of SNRs. In general, the agreement is good whenever the estimation errors are small, which can happen for a very low SNR if the data collection time is long. The usefulness of the derived formulas is further illustrated by calculating the theoretical MSE for
varying operating regimes. In a studied scenario, it is found that for moderate clutter, the
excess error due to noise in the RC can only be neglected if the SNR in the RC is about 60 dB higher than that in the SC. 
However, as the power of the clutter interference is increased beyond a certain critical value, the requirement on the SNR in the RC relative to that in the SC becomes even more strict. This is natural since the ability to cancel clutter depends on the accuracy of the interference subspace estimate, and the requirement should become stricter with increasing clutter power.
Similar calculations can be made to investigate the effect of varying other system parameters 
without having to resort to costly Monte-Carlo simulations. Our analysis is therefore useful for guiding the design and practical experimental setup of a passive radar system. 

\appendices

\section{Covariance Matrix of the Gradient}
\label{Appendix1}

The purpose of this appendix is to derive the term $\Wbf = \E[V'(\thetabf_0)V'^T(\thetabf_0)]$ appearing in (\ref{eq:AsCovdef}).
Clearly, from (\ref{eq:GlobalML}), ignoring the weighting since $\sigma_{e,k}^2 \overset{\Delta}{=}\sigma_e^2$, we have
\begin{equation}
V'(\thetabf_0) =
 \sum_{k=1}^K P'_k(\thetabf_0)\, ,
 \label{eq:Vprim}
\end{equation}
where we regard $P_k$ as a function of $\thetabf$, indirectly through the known functions $\tau_k(\thetabf)$ and $\omega_k(\thetabf)$. Alternatively, for a single node we omit the sum and let $\thetabf = (\tau,\omega)^T$ be the parameter vector. Let $P'_k$ denote the derivative of $P_k$ w.r.t. one of the components in $\thetabf$. 
From (\ref{eq:ML4}) we have
\begin{equation}
P_k'(\thetabf_0) = \ybf_k^H 
\hat{\Pbf}'_k
\ybf_k^{}\, .
\end{equation}
The projection matrix $\hat{\Pbf}_k$ depends on $\thetabf$ through the steering vector $\hat{\abf}$, where the dependence on $\tau_k$ and $\omega_k$ has been suppressed.
Applying the perturbation theory of projection matrices (e.g. \cite{GolubP:73,viberg_ottersten_1991}) now yields
\begin{align}
P_k' &= \ybf_k^H \left(
\hat{\Pbf}_k^{\perp}\hat{\Pibf}_k^{\perp} \hat{\abf}' (\hat{\Pibf}_k^{\perp}\hat{\abf})^+ + 
(\hat{\Pibf}_k^{\perp}\hat{\abf})^{+H} \hat{\abf}'^H\hat{\Pibf}_k^{\perp} \hat{\Pbf}_k^{\perp}
\right) \ybf_k^{} 
\nonumber \\
&= 2\,\Re\left\{ 
\ybf_k^H
\hat{\Pbf}_k^{\perp}\hat{\Pibf}_k^{\perp} \hat{\abf}' (\hat{\Pibf}_k^{\perp}\hat{\abf})^+
\ybf_k^{} \right\} ,
\label{eq:Pkprim}
\end{align}
where $(\cdot)^+$ denotes the Moore-Penrose pseudo-inverse of a matrix and $\hat{\Pbf}_k^{\perp}=\Ibf-\hat{\Pbf}_k$.
In these expressions, it is understood that $\hat{\abf}$ and its derivative $\hat{\abf}' =
\partial\hat{\abf}/\partial\thetabf$ should be evaluated at the true parameter vector $\thetabf_0$.
Next, we insert the expression for the received SC signal as
\begin{align}
\ybf_k &= \ybf_{I,k} + d_k\abf + \ebf_k\,
\label{eq:ydef} \\
\ybf_{I,k} &=  [\sbf_k,\Sbf_k] \left[
\begin{array}{l}
   b_{k}  \\ \cbf_k 
\end{array} \right] = \Sbf_{I,k} \fbf_k
\end{align}
where $\ybf_{I,k}$ is the total interference component. In addition, a first-order expression for the interference cancellation matrix $\hat{\Pibf}_k^{\perp}$ is needed. Arrange the RC signal and noise components conformably with $\Sbf_{I,k}$ above, so that $\Xbf_{I,k} = a_k\,\Sbf_{I,k} + \Nbf_{I,k}$.  
Assuming $\|\Nbf_{I,k}\|$ to be much smaller than $\|a_k\Sbf_{I,k}\|$ (in the mean-square sense), we have
\begin{equation}
\hat{\Pibf}_k^{\perp} \simeq \Pibf_k^{\perp} - \Delta\Pibf_k \, ,
\end{equation}
where $\Pibf^{\perp}_k$ is defined in (\ref{eq:Piperp}). The $\simeq$ sign means first-order approximation in mean-square, and using \cite{GolubP:73} we get
\begin{equation}
\Delta\Pibf_k = \Pibf_k^{\perp}
\Nbf_{I,k}^{}(a_k\Sbf_{I,k})^+ +
(a_k\Sbf_{I,k})^{+H} \Nbf_{I,k}^H \Pibf_k^{\perp} \ .
\label{eq:DeltaPi}
\end{equation}
In the approximation of (\ref{eq:Pkprim}), we will use that $\Pibf_k^\perp\ybf_{I,k} = 0$ and $({\Pibf}_k^{\perp}\abf)^+\abf = 1$. Further, let
$$
{\Pbf}_k= 
\frac{{\Pibf}_k^{\perp}{\abf}{\abf}^H{\Pibf}_k^{\perp}}{{\abf}^H{\Pibf}_k^{\perp}{\abf}} = \Ibf - \Pbf_k^{\perp}
$$
denote the projection matrix (\ref{eq:Pdef}) onto the noise-free interference-canceled steering vector $\Pibf_k^{\perp}\abf$. Both these factors are subject to perturbation, where $\Delta\Pibf_k$ is given in (\ref{eq:DeltaPi}). Since $\hat{\abf}$ is constructed from $\xbf_k$ and not $\sbf_k$, we have $\hat{\abf}\approx a_k\abf$ and $\Delta\abf=\hat{\abf}-a_k\abf=\nbf_k\odot \vbf(\omega_k)$.
The first-order expression of $\hat{\Pbf}_k$ is therefore obtained as
\begin{align}
\hat{\Pbf}_k &= \Pbf_k + \Delta\Pbf_k \\
\Delta\Pbf_k &= \Pbf_k^{\perp} (\Pibf_k^{\perp}\Delta \abf - \Delta\Pibf_k\, a_k\abf)( \Pibf_k^{\perp}a_k\abf)^+ \\ & + (\Pibf_k^{\perp}a_k\abf)^{+H}
(\Pibf_k^{\perp}\Delta \abf - \Delta\Pibf_k\, a_k\abf)^H \Pbf_k^{\perp} .
\label{eq:DeltaP}
\end{align}
We note, in passing, that both $\Delta\Pibf_k$ and $\Delta\Pbf_k$ are of order $\sigma_n^2/|a_k|^2$ in the mean-square sense, which we recognize as the Noise-to-Signal Ratio in the RC, as expected.

Inserting (\ref{eq:ydef}) -- (\ref{eq:DeltaP}) into (\ref{eq:Pkprim}), we can split the gradient into the two components
\begin{equation}
P_k'(\thetabf_0) \simeq T_{1,k} + T_{2,k} \,
\end{equation}
where $T_{1,k}$ represents the first-order contribution from the noise in the SC, and $T_{2,k}$ contains the effect of the noise in the RC channel. After some manipulations, which we omit, the terms are obtained as
\begin{equation}
T_{1,k} = 2\,\Re\left\{
\ebf_k^H \Pbf_k^\perp \Pibf_k^{\perp} \abf' d_k
\right\}
\end{equation}
and
\begin{equation}
T_{2,k} = \frac{2}{a_k}\,\Re\left\{
\left(\Nbf_{I,k}\fbf_{I,k} + \Delta\abf\,d_k\right)^H
\Pbf_k^\perp\Pibf_k^{\perp}\abf' d_k \right\}
\end{equation}
respectively. 

Given the first-order approximation of the gradient, we can now compute its approximate covariance matrix. Since the terms in (\ref{eq:GlobalML}) are independent, the covariance matrix can be computed for each term separately. Denote the first-order covariance matrix for the $k$-th node $\Wbf_k$. Its $(i,j)$-th element is then given by
\begin{equation}
\Wbf_k(i,j) = \E \left[ 
\frac{\partial P_k(\thetabf)} {\partial\theta_i}
\frac{\partial P_k(\thetabf)} {\partial\theta_j}
\right] = T_{1,k}(i,j) + T_{2,k}(i,j)
\label{eq:Qdef}
\end{equation}
where 
we have used that $T_{1,k}$ and $T_{2,k}$ are independent.
Denoting
${\abf}'_i =
\partial{\abf}/\partial\theta_i$,
we have
\begin{equation}
\label{eq:T1ij}
\begin{aligned}
T_{1,k}(i,j)
&= 4\, \E\!\left[
  \Re\!\left\{\ebf_k^H \Pbf_k^\perp \Pibf_k^{\perp} \abf_i' d_k\right\}
  \,\Re\!\left\{\ebf_k^H \Pbf_k^\perp \Pibf_k^{\perp} \abf_j' d_k\right\}
\right]
\end{aligned}
\end{equation}
and 
\begin{equation}
\label{eq:T2ij}
\begin{split}
T_{2,k}(i,j) = \frac{4}{|a_k|^2}\,\E\!\biggl[
\Re\!\left\{\!\left(\Nbf_I\fbf_I+\Delta\abf\,d_k\right)^{\!H}
  \Pibf_k^{\perp}\Pbf_k^{\perp}\abf_i' d_k\right\}
\biggr.\\[-0.3ex]
\biggl.\cdot\Re\!\left\{\!\left(\Nbf_I\fbf_I+\Delta\abf\,d_k\right)^{\!H}
  \Pibf_k^{\perp}\Pbf_k^{\perp}\abf_j' d_k\right\}
\biggr]
\end{split}
\end{equation}
where $\Delta \abf = \nbf_k\odot \vbf(\omega_k)$. 
Next, using the formula
$$
2\,\Re\{u\}\Re\{v\} = \Re\{u v\} + \Re\{u v^H\}
$$
for two complex scalars $u$ and $v$
and that $\ebf_k$ is circularly symmetric, (\ref{eq:T1ij}) becomes
\begin{align}
 T_{1,k}(i,j) &= 2\,|d_k|^2 \E 
\left[ \Re \left\{ \abf_i'^H\Pibf_k^{\perp} \Pbf_k^{\perp} \ebf_k^{}\ebf_k^H  
\Pbf_k^{\perp}\Pibf_k^{\perp} \abf'_j \right\} \right] 
\nonumber \\
&= 2\,|d_k|^2\sigma_e^2\, \Re \left\{ \abf_i'^H \Tilde{\Pbf}_k\abf'_j \right\}\,,
\label{eq:T1ijfinal}
\end{align}
where $\Tilde{\Pbf}_k$ is defined in (\ref{eq:Ptilde}).

Next, to evaluate $T_{2,k}(i,j)$, we recall that $\Nbf_{I,k}=[\nbf_k\, ,\, \Nbf_k]$. Thus, we can rewrite the noise contribution as
\begin{align*}
\Nbf_{I,k}\fbf_k + \Delta\abf\,d_k &= \Nbf_k\cbf_k + \nbf_k b_k + \nbf_k \odot \vbf(\omega_k) \,d_k\\
&= 
\Zbf_k \text{vec}(\Nbf_{I,k})
\, ,
\end{align*}
where $\Zbf_k$ is defined in (\ref{eq:Zdef}).
The matrix $\Nbf_{I,k}$ is a Toeplitz matrix of the RC noise samples. Thus, we can write
$$
\text{vec}(\Nbf_{I,k}) = \Jbf_k \nbf_{I,k}\, ,
$$
where $\nbf_{I,k} = [n(-L),n(-L+1,\dots,n(N-1)]^T$ is the vector of all noise samples, and where $\Jbf_k$ is a selection matrix that picks out the respective component from $\nbf_{I,k}$. This leads to
\begin{equation}
\Nbf_I\fbf_I + \Delta\abf\,d_k = \Zbf_k\Jbf_k\nbf_{I,k} \, .
\label{eq:T2noise}
\end{equation}
Inserting (\ref{eq:T2noise}) into (\ref{eq:T2ij}) now leads to
\begin{equation}
 T_{2,k}(i,j) = 2\,|d_k|^2\frac{\sigma_n^2}{|a_k|^2}\, \Re \left\{ \abf_i'^H\Tilde{\Pbf}_k\Zbf_k^{}\Jbf_k^{}\Jbf_k^{T}\Zbf_k^{H}\Tilde{\Pbf}_k^{} \abf'_j \right\}\,.
\label{eq:T2ijfinal}
\end{equation}
Collecting the partial derivatives $\abf'_i$ into the matrix \eqref{eq:Ddef}
and inserting (\ref{eq:T1ijfinal}) and (\ref{eq:T2ijfinal}) into (\ref{eq:Qdef}), we can express the asymptotic covariance matrix of the gradient as
\begin{equation}
\begin{split}
\Wbf_k &= 2\,|d_k|^2\sigma_e^2\, \Re \left\{ \Dbf^H \Tilde{\Pbf}_k\Dbf\right\}  \\ &+
2\,|d_k|^2\frac{\sigma_n^2}{|a_k|^2}\, \Re \left\{ \Dbf^H\Tilde{\Pbf}_k\Zbf_k^{}\Jbf_k^{}\Jbf_k^{T}\Zbf_k^{H}\Tilde{\Pbf}_k^{} \Dbf
\right\}\,.
\label{eq:Qfinal}
\end{split}
\end{equation} 
The two terms represent the contributions from the noise in the SC and in the RC respectively.
Finally, adding all terms as in (\ref{eq:Vprim}), we arrive at
\begin{equation}
\Wbf = \E[V'(\thetabf_0)V'^T(\thetabf_0)] =
\sum_{k=1}^K \Wbf_k\, .
\label{eq:Vprimcov}
\end{equation}

\section{Asymptotic Hessian Matrix}
\label{Appendix2}

For the asymptotic Hessian matrix we set \( \sigma_n ^2 = \sigma_e^2 = 0 \) and differentiate the projection matrix in \eqref{eq:ML4} accordingly. We set \( \bbf = \Pibf^\perp \abf\) and follow the calculations outlined in Appendix B of \cite{viberg_ottersten_1991}. As above, we use the notation \(  \partial_i \bbf/\partial\theta_i = \bbf_i ' \), and further \(  \partial^2 \bbf/\partial\theta_i\partial\theta_j = \bbf_{ij} '' \) . We have already seen in
\eqref{eq:Pkprim} that 
\begin{equation}
\Pbf_i ' = \bbf_i ' \bbf^{+} + \bbf (\bbf^{+}_i) ' = \Pbf^\perp  \bbf_i ' \bbf^{+} + (\Pbf^\perp \bbf_i ' \bbf^{+})^{H} ,
\label{eq:Pprim_rewr}
\end{equation}
while the second derivatives are obtained as in \cite{viberg_ottersten_1991}
\begin{equation} \begin{split}
\Pbf_{ij} ''  = & - \Pbf^\perp \bbf_j ' \bbf^{+} \bbf_i ' \bbf^{+} - (\bbf^{+})^{H} (\bbf_j ')^H \Pbf^\perp \bbf_i ' \bbf^{+} + \Pbf^\perp \bbf_{ij} '' \bbf^{+} \\ &  + \Pbf^\perp \bbf_i ' (\bbf_j ' )^{H}\Pbf^\perp / \|\bbf\|^2 - \Pbf^\perp \bbf_i ' \bbf^{+} \bbf_i ' \bbf^{+} + (\dots)^H .
\end{split}
\label{eq:Psec_rewr}
\end{equation}
where \( (\dots)^H \) means the Hermitian transpose of all the previous terms. Now, 
computing 
\( \ybf_k^{H} \Pbf_{ij} '' \ybf_k  \)
we observe that \(\bbf_i ' \) still belongs to \( \text{span}(\Pibf^\perp) \), because \(\Pibf^\perp \) does not depend on the parameters \( \thetabf \). Moreover, we recall that \( \ybf_k = \gbf_k + d_k \abf_k \) for some \( \gbf_k \in \text{span}(\Sbf_I) \); thus \( \ybf_k ^H \Pbf^\perp = \ybf_k ^H(\Ibf - \Pbf) = \ybf_k ^H - (\Pbf \ybf_k)^H = \gbf_k ^H \) and by \( \gbf_k \perp \bbf_i ' \), we can conclude that all the factors in \eqref{eq:Psec_rewr} beginning with \( \Pbf^\perp \bbf_i '  \) vanish when left-multiplied by \( \ybf_k ^H \). 
Therefore, using \(\bbf^+ = \bbf^H / \|\bbf\|^2 \) and \( \bbf^H \abf = \|\bbf\|^2 \), we obtain \begin{equation}
\ybf_k ^H \Pbf_{ij} '' \ybf_k = -2 |d_k|^2 \Re \left\{ (\bbf_j ')^H \Pbf^\perp  \bbf_i ' \right\}\,.
\label{eq:Vsec_rewr}
\end{equation}
Note that (\ref{eq:Vsec_rewr}) equals (\ref{eq:T1ijfinal}) up to a multiplicative factor as expected. The above can be put into matrix form as
\begin{equation}
\Hbf_k = -2\,|d_k|^2\, \Re \left\{ \Dbf^H _k \Tilde{\Pbf} _k^{} \Dbf _k^{} \right\}\, ,
\label{eq:Hfinal}
\end{equation}
which proves \eqref{eq:CRB_theta}.

\section{Consistency discussion}
\label{Appendix3}
Before proving Theorem \ref{thm:rf_unambiguous}, we need to recall a technical result from \cite{azais_wschebor_2009}:

\begin{proposition}[6.11 in \cite{azais_wschebor_2009}] \label{prop:azaiswschebor}
Let \(\mathcal{Y} = \{ Y(t) \, : \, t \in W  \}  \) be a random field with values in \( \mathbb{R}^{m + k} \) and \(W \subseteq \mathbb{R}^d \) open, \( k \ge 1\) . Let \( u \in \mathbb{R}^{m + k} \) and \(I\) a subset of \(W\). We assume that \( \mathcal{Y} \) satisfies the following conditions:

\begin{itemize}
    \item the paths \( t \rightsquigarrow Y(t) \) are of class \( \mathcal{C}^1 \),
    \item for each \(t \in W \), the random vector \( Y(t) \) has a density, and there exists a constant \( C \) such that \( p_{Y(t)}(x) \le C \) for \( t \in I \) and \(x\) is some neighborhood of \(u\),
    \item the Hausdorff dimension of \(I\) is smaller or equal than \(m\).
\end{itemize}
Then, almost surely, there is no point \( t \in I \) such that \( Y(t) = u \).

\end{proposition}

\begin{proof}[Proof of Theorem \ref{thm:rf_unambiguous}]
We may assume without loss of generality that \( \tau_0 = \omega_0 = 0 \), and the excluded set then becomes \( T_1 = \{0\} \cup \bigcup_{p \ne q} \{t_p - t_q\} \).

We define \( T \coloneqq T_3 \times \mathbb{R} \) with \( T_3 \subseteq \mathbb{R} \setminus T_1 \) a compact set. Since \( \abf(0,0) \ne \mathbf{0} \) w.p.\ \(1\), the event whose probability shall be assessed is equivalent to \( \{\exists \, (\tau, \omega, \alpha) \in T \times \mathbb{C} \, : \, \abf(\tau,\omega) + \alpha\, \abf(0,0) = 0\} \). Consider then the Gaussian random field \( \hbf : W \to \mathbb{R}^{2N} \) on the open set \( W \coloneqq (\mathbb{R}\setminus T_1) \times \mathbb{R}^3 \), whose \(2N\) components are the real and imaginary parts of the entries of \( \abf(\tau,\omega) + \alpha\, \abf(0,0) \), parametrized by \( (\tau, \omega, \alpha_1 , \alpha_2) \), under the identification \( \alpha = \alpha_1 + j \alpha_2 \cong (\alpha_1, \alpha_2) \). If the hypotheses of Proposition \ref{prop:azaiswschebor} for \( \hbf \) on \( W \), with \( u = \mathbf{0} \) and \( I = T \times \mathbb{R}^2 \subset W \) are verified, then the level set \( \{ \hbf = \mathbf{0} \} \) is almost surely empty over \( I \), which is our thesis.

We factor \( \hbf(\tau, \omega, \alpha) = \Bbf(\omega, \alpha)\, \xbf(\tau) \), where \( \xbf(\tau) \in \mathbb{R}^{4N} \) collects the real and imaginary parts of \( s(t_1-\tau),\dots,s(t_N-\tau),\, s(t_1),\dots,s(t_N) \), and \( \Bbf(\omega, \alpha) = \begin{pmatrix} \Rbf (\omega) & \Abf(\alpha) \end{pmatrix} \) is the \( 2N \times 4N \) real matrix whose blocks \( \Rbf(\omega) \) and \( \Abf(\alpha) \) are \( 2N \times 2N \) block-diagonal with \( 2\times 2 \) diagonal blocks \( \left(\begin{smallmatrix} \cos\omega t_n & -\sin\omega t_n \\ \sin\omega t_n & \cos\omega t_n \end{smallmatrix}\right) \) and \( \left(\begin{smallmatrix} \alpha_1 & -\alpha_2 \\ \alpha_2 & \alpha_1 \end{smallmatrix}\right) \) respectively, \( n=1,\dots,N \). As a consequence of Theorem 11 in \cite{belyaev_1960}, the sample paths of \( \xbf \) are of class \( \mathcal{C}^\infty \) w.p.\(1\) in \( \tau \), and so are the sample paths of \(\hbf \) in \( (\tau, \omega, \alpha) \).

Continuing further: a direct computation gives \( \Bbf\Bbf^T = (1 + |\alpha|^2)\, \Ibf_{2N} \), so from \( \boldsymbol{\Sigma}_{\hbf} = \Bbf\, \boldsymbol{\Sigma}_{\xbf}\, \Bbf^T \) we obtain, for every \( \cbf \in \mathbb{R}^{2N} \),
\begin{equation} \label{eq:Poinc} \begin{split}
\cbf^T \boldsymbol{\Sigma}_{\hbf}(\tau,\omega,\alpha)\, \cbf & = (\Bbf^T \cbf)^T \boldsymbol{\Sigma}_{\xbf}(\tau)\, \Bbf^T \cbf \\ & \ge (1 + |\alpha|^2)\, \lambda_{\min}\big(\boldsymbol{\Sigma}_{\xbf}(\tau)\big) \|\cbf\|_2^2, \end{split}
\end{equation}
hence \( \lambda_{\min}(\boldsymbol{\Sigma}_{\hbf}(\tau,\omega,\alpha)) \ge \lambda_{\min}(\boldsymbol{\Sigma}_{\xbf}(\tau)) \) for all \( (\omega, \alpha) \).

It remains to show \( \boldsymbol{\Sigma}_{\xbf}(\tau) \) is non-singular for every \( \tau \in \mathbb{R}\setminus T_1 \). With the auxiliary vector \( \zbf(\tau) = \big(s(t_1-\tau),\dots,s(t_N-\tau),\, s(t_1),\dots,s(t_N)\big)^T \in \mathbb{C}^{2N} \) and \( \boldsymbol{\Sigma}_{\zbf}(\tau) = \mathbb{E}[\zbf(\tau)\zbf(\tau)^{\mathrm H}] \), the vector \( \xbf(\tau) \) equals \( \big(\mathfrak{Re}(\zbf)^T, \mathfrak{Im}(\zbf)^T\big)^T \) up to a permutation; since \( s \) is circularly symmetric, \( \boldsymbol{\Sigma}_{\xbf}(\tau) \) is non-singular iff \( \boldsymbol{\Sigma}_{\zbf}(\tau) \) is. The latter, being a covariance matrix, is positive semi-definite, so it suffices to show \( \cbf^{\mathrm H} \boldsymbol{\Sigma}_{\zbf}(\tau)\, \cbf = 0 \Rightarrow \cbf = \mathbf{0} \). By the Wiener–Khinchine theorem, \( \mathbb{E}[\zbf(\tau)_n \overline{\zbf(\tau)_\ell}] = \int_\mathbb{R} e^{j f (u_n - u_\ell)} \mathcal{S}_s(f)\, df \) with \( u_n \in \{t_1 - \tau, \dots, t_N - \tau, t_1, \dots, t_N\} \), so
\begin{equation} \label{eq:WienerKhinchine} \begin{split}
\cbf^{\mathrm H} \boldsymbol{\Sigma}_{\zbf}(\tau)\, \cbf & = \int_\mathbb{R} \Big( \sum_n \overline{c_n} e^{j f u_n} \Big) \Big( \sum_\ell c_\ell e^{-j f u_\ell} \Big) \mathcal{S}_s(f)\, df \\ & = \int_{-B/2}^{B/2} \Big| \sum_{n=1}^{2N} c_n e^{-j f u_n} \Big|^2 df , \end{split}
\end{equation}
the last equality using \( \mathcal{S}_s = \mathbf{1}_{[-B/2,B/2]} \). For \( \tau \notin T_1 \) the \( u_n \) are pairwise distinct - the only possible collisions \( t_p - \tau = t_q \) occur for \( \tau \in T_1 \) - so the exponentials \( \{e^{-jfu_n}\}_{n=1}^{2N} \) are linearly independent over \( \mathbb{C} \); for \( \cbf \ne \mathbf{0} \) the entire function \( f \mapsto \sum_n c_n e^{-jfu_n} \) is not identically zero and vanishes at finitely many points of \( [-B/2,B/2] \) only, making \eqref{eq:WienerKhinchine} strictly positive. In particular, by \eqref{eq:Poinc}, \( \hbf \) has a density at every point of \( W \).

Finally, \( \tau \mapsto \lambda_{\min}(\boldsymbol{\Sigma}_{\xbf}(\tau)) \) is continuous and strictly positive on \( \mathbb{R}\setminus T_1 \), hence attains a minimum \( c_{T_3} > 0 \) on the compact set \( T_3 \); with \eqref{eq:Poinc} this gives \( \lambda_{\min}(\boldsymbol{\Sigma}_{\hbf}) \ge c_{T_3} \) on \( I \), so that
\[
p_{\hbf(\tau,\omega,\alpha)}(x) \le (2\pi)^{-N} \det(\boldsymbol{\Sigma}_{\hbf})^{-1/2} \le (2\pi)^{-N} c_{T_3}^{-N}
\]
for all \( x \in \mathbb{R}^{2N} \) and \( (\tau,\omega,\alpha) \in I \), fulfilling the second hypothesis of Proposition \ref{prop:azaiswschebor}. Since \( \dim_H(I) \le \dim_H(T_3) + 3 \le 4 \), the proposition applies with \( m = 4 \) and \( k = 2N - 4 \ge 2 \) and \( \{\hbf = \mathbf{0}\} \) is almost surely empty on \( I \).
\end{proof}

\bibliographystyle{ieeetr}
\bibliography{DistributedRadarRefs}

\end{document}

%% file: macros1.tex
\def\psfancypar#1#2{\begingroup\def\par{\endgraf\endgroup\lineskiplimit=0pt}
               \setbox2=\hbox{\large\sc #2}
               \newdimen\tmpht \tmpht \ht2 \advance\tmpht by \baselineskip
               \font\hhuge=cmti12 at \tmpht
              \setbox1=\hbox{{\hhuge #1}}
               \count7=\tmpht \count8=\ht1
               \divide\count8 by 1000 \divide\count7 by \count8
               \tmpht=.001\tmpht\multiply\tmpht by \count7
               \font\hhuge=cmbx10 at \tmpht
               \setbox1=\hbox{{\hhuge #1}}
               \noindent
                \hangindent1.05\wd1
               \hangafter=-2 {\hskip-\hangindent
               \lower1\ht1\hbox{\raise1.0\ht2\copy1}%
                \kern-0\wd1}\copy2\lineskiplimit=-1000pt}

\newcommand{\boxit}[1]{\vbox{\hrule\hbox{\vrule\kern6pt
\vbox{\kern6pt#1 \vspace{-10 pt}\kern6pt}\kern6pt\vrule}\hrule}}
\newcommand{\boxitsl}[1]{\vbox{\hrule\hbox{\vrule\kern6pt
\vbox{\kern6pt#1 \vspace{3 pt}\kern6pt}\kern6pt\vrule}\hrule}}
\def\thetabf{{\mbox{\boldmath$\theta$\unboldmath}}}

\def\etabf{\mbox{\boldmath$\eta$\unboldmath}}

\newcommand{\E}{\mbox{{\rm E}}}

\def\reals{ { {\rm  I \kern-0.15em R }  } }
\def\complex{\hbox{ \,{{\rm C} \kern-0.50em \raise0.20ex {  |}}\,}}

\def\Pibf{\hbox{$\boldsymbol{\Pi}$}}

\def\abf{\hbox{\bf a}}
\def\bbf{\hbox{\bf b}}
\def\cbf{\hbox{\bf c}}

\def\ebf{\hbox{\bf e}}
\def\fbf{\hbox{\bf f}}
\def\gbf{\hbox{\bf g}}
\def\hbf{\hbox{\bf h}}

\def\nbf{\hbox{\bf n}}

\def\rbf{\hbox{\bf r}}
\def\sbf{\hbox{\bf s}}

\def\ubf{\hbox{\bf u}}
\def\vbf{\hbox{\bf v}}

\def\xbf{\hbox{\bf x}}
\def\ybf{\hbox{\bf y}}
\def\zbf{\hbox{\bf z}}
\def\rbf{\hbox{\bf r}}
\def\xbf{\hbox{\bf x}}
\def\ybf{\hbox{\bf y}}
\def\Abf{\hbox{\bf A}}
\def\Bbf{\hbox{\bf B}}
\def\Cbf{\hbox{\bf C}}
\def\Dbf{\hbox{\bf D}}

\def\Hbf{\hbox{\bf H}}
\def\Ibf{\hbox{\bf I}}
\def\Jbf{\hbox{\bf J}}

\def\Nbf{\hbox{\bf N}}

\def\Pbf{\hbox{\bf P}}
\def\Qbf{\hbox{\bf Q}}
\def\Rbf{\hbox{\bf R}}
\def\Sbf{\hbox{\bf S}}

\def\Wbf{\hbox{\bf W}}
\def\Xbf{\hbox{\bf X}}

\def\Zbf{\hbox{\bf Z}}

\def\be{\vskip .3cm \begin{equation}}
\def\ee{\end{equation} \vskip .4cm \noindent}

\newcounter{remarknr}
\renewcommand{\theremarknr}{\arabic{remarknr}}
\newenvironment{remark}{\vskip\baselineskip
\refstepcounter{remarknr}\noindent{\bf
Remark~\theremarknr}\hskip .8em}{\ \hfill $\Box$ \vskip\baselineskip}
\newtheorem{theorem}{Theorem}

\newcounter{assumpnr}
\renewcommand{\theassumpnr}{\arabic{assumpnr}}



%% file: exmacros.tex
\newcounter{examplenr} 
\renewcommand{\theexamplenr}{\arabic{examplenr}}

